\documentclass{article}
\usepackage{mathtools,amssymb,amsthm,tikz,tikz-cd,multirow}
\usepackage{fullpage}
\usepackage{etoolbox}
\usepackage{extarrows}
\usepackage{enumitem}
\usepackage{textcomp}
\usepackage{makecell}
\usepackage{tensor}
\usepackage[title]{appendix}
\numberwithin{equation}{section}
\usetikzlibrary{calc}
\usepackage[english]{babel}
\usepackage[utf8]{inputenc}
\usepackage{amsmath}
\usepackage{amsfonts}
\usepackage{graphicx}
\usepackage[colorinlistoftodos]{todonotes}
\usepackage{blindtext}
\usepackage{amssymb}
\usepackage{mathtools,tikz,tikz-cd,multirow}
\usepackage{bbm}
\usepackage{amsthm}

\newtheorem{theorem}{Theorem}[section]

\newtheorem{proposition}[theorem]{Proposition}
\theoremstyle{remark}
\newtheorem*{remark}{Remark}
\usepackage{mathtools,amssymb,amsthm,tikz,tikz-cd,multirow}

\definecolor{darkgreen}{RGB}{0,180,0}
\colorlet[named]{green}{darkgreen}
\theoremstyle{definition}

\newcommand{\caps}[2]{\begin{tikzpicture}
  \draw (0,-0.5) arc(-90:90:0.5);
  \filldraw[black] (0.5,0) circle (2pt);
  \node at (0,-0.5) [anchor=east] {$#2$};
  \node at (0,0.5) [anchor=east] {$#1$};
\end{tikzpicture}}

\newcommand{\gammaice}[4]{\begin{tikzpicture}
\coordinate (a) at (-.75, 0);
\coordinate (b) at (0, .75);
\coordinate (c) at (.75, 0);
\coordinate (d) at (0, -.75);
\coordinate (aa) at (-.75,.5);
\coordinate (cc) at (.75,.5);
\draw (a)--(c);
\draw (b)--(d);
\draw[fill=white] (a) circle (.25);
\draw[fill=white] (b) circle (.25);
\draw[fill=white] (c) circle (.25);
\draw[fill=white] (d) circle (.25);
\node at (0,1) { };
\node at (a) {$#1$};
\node at (b) {$#2$};
\node at (c) {$#3$};
\node at (d) {$#4$};
\path[fill=white] (0,0) circle (.2);
\node at (0,0) {$z_i$};
\end{tikzpicture}}

\title{Computation of partition functions of free fermionic solvable lattice models via permutation graphs}

\author{Chenyang Zhong\thanks{Department of Statistics, Columbia University}}

\date{\today}

\begin{document}
\maketitle

\begin{abstract}
In this paper, we introduce a novel and general method for computing partition functions of solvable lattice models with free fermionic Boltzmann weights. The method is based on the ``permutation graph'' and the ``$F$-matrix'': the permutation graph is a generalization of the $R$-matrix, and the $F$-matrix is constructed based on the permutation graph. The method allows generalizations to lattice models that are related to Cartan types B and C. Two applications are presented: they involve an ice model related to Tokuyama's formula and another ice model representing a Whittaker function on the metaplectic double cover of $\mathrm{Sp}(2r,F)$ with $F$ being a non-archimedean local field. 
\end{abstract}

\section{Introduction}

Solvable lattice models, which originate from statistical physics, have been playing an important role in various areas of mathematics and mathematical physics, including algebraic combinatorics (e.g. \cite{BBBG,BSW,Kup,Kup2,WZ}), quantum field theory (e.g. \cite{Bax3,BPZ,DMS}), and integrable probability (e.g. \cite{BP,BW,CP,OP}). A solvable lattice model is usually based on a finite lattice, and a state of the model is a labeling of the edges of the lattice. Each state of the model is then associated with a locally determined Boltzmann weight. 

The partition function of a lattice model, which is defined as the sum of the Boltzmann weights of all admissible states of the model, is of great importance. From the point of view of statistical physics, the partition function encodes thermodynamic properties of the lattice model. Recent works (see e.g. \cite{Bor,BBF,MS2}) have also found realizations of various symmetric functions--such as the Schur, Hall-Littlewood, and Grothendieck polynomials--as partition functions of certain solvable lattice models.

The Yang-Baxter equation (cf. \cite{Jim,Maj}), also known as the ``star-triangle relation'', plays a key role in solvable lattice models. It reveals symmetries of the partition function of the lattice model. In Baxter's seminal work \cite{Bax2,Bax1} and many later works, the Yang-Baxter equation is crucially utilized to obtain explicit expressions for partition functions of solvable lattice models. 

However, going from the Yang-Baxter equation to an explicit evaluation of the partition function often requires quite a bit of non-trivial work, such as the combinatorics of Gelfand-Tsetlin patterns and Proctor patterns (see e.g. \cite{BBF,Iva}) and the Izergin-Korepin technique (see e.g. \cite{Ize,Kor,Whe}). A natural question is, for a given lattice model that satisfies the Yang-Baxter equation, is there a \underline{general} route to the computation of the partition function? This paper provides an answer for a class of solvable lattice models called ``free fermionic solvable lattice models'', which we will review in Section \ref{Sect.1.1}. A high-level overview of the method will be given in Section \ref{Sect.1.2}. 

\subsection{Free fermionic solvable lattice models}\label{Sect.1.1}

Recently, there has been a series of works that relate solvable lattice models to representation theory. This originates in the seminal work \cite{BBF}, where a parametrized Yang-Baxter equation with non-abelian parameter group is introduced. The Yang-Baxter equation corresponds to a six-vertex model with \underline{free fermionic} Boltzmann weights. ``Free fermionic'' means that, if we denote by $a_1,a_2,b_1,b_2,c_1,c_2$ the Boltzmann weights of the six-vertex model (see Section \ref{Sect.2} for details), then the constraint $a_1a_2+b_1b_2-c_1c_2=0$ is satisfied. The partition function of the six-vertex model is shown to be equal to the product of a Schur polynomial and a deformation of the Weyl denominator of the general linear group, which provides an alternative proof of Tokuyama's formula \cite{Tok}. The result of \cite{BBF} is generalized to factorial Schur functions in \cite{BMN}. Later works \cite{BBBG3,BBBG4,BBBG5,BBCFG} construct solvable lattice models whose partition functions represent metaplectic Whittaker functions and Iwahori Whittaker functions. These culminate in the work \cite{BBBG2} that constructs a supersymmetric solvable lattice model whose partition function gives metaplectic Iwahori Whittaker functions. 

The above developments are for Cartan type A. There is also a parallel line of works for other Cartan types. Hamel and King \cite{HK1,HK2} and Ivanov \cite{Iva} constructed lattice models whose partition functions equal the product of an irreducible character and a deformation of Weyl's denominator of the symplectic group $\mathrm{Sp}(2n,\mathbb{C})$. The Yang-Baxter equation as developed in \cite{BBF} is used in Ivanov's work. Brubaker et al. \cite{BBCG} constructed a solvable lattice model and made the conjecture that its partition function represents the metaplectic Whittaker function on the double cover of $\mathrm{Sp}(2n,F)$, where $F$ is a non-archimedean local field. In later works \cite{Mot,MSW}, a dual version of the model in \cite{Iva} and generalizations of the models in \cite{BBCG,Iva} are studied. In \cite{Gra}, Ivanov's model is generalized to metaplectic ice for Cartan type C. Further developments include \cite{BS,BS2,BSn}.

The lattice model in \cite{BBF} is based on a finite rectangular lattice. Each state of the model is represented by a labeling of the edges of the lattice by $\pm$ signs, which we refer to as ``spins'' in the following. Given a labeling, each vertex of the lattice is associated with a Boltzmann weight determined by the spins on its adjacent edges. The Boltzmann weight of the corresponding state is defined as the product of the Boltzmann weights of all the vertices. The lattice models in \cite{Iva} and \cite{BBCG} are also based on a finite rectangular lattice, but involve two types of vertices called $\Delta$ ice and $\Gamma$ ice that alternate in rows. There are also U-turn vertices, which we call ``cap vertices'' in this paper, that connect two adjacent rows of $\Delta$ ice and $\Gamma$ ice on the right boundary. The Boltzmann weights for these vertices will be reviewed in Section \ref{Sect.2}. The lattice models in \cite{BBF} and \cite{BBCG} will be reviewed in Sections \ref{Sect.4} and \ref{Sect.6}, respectively.

\subsection{Overview of the strategy}\label{Sect.1.2}

In this subsection, we give a high-level overview of our strategy for computing partition functions of free fermionic solvable lattice models. We focus on the lattice models in \cite{BBF} and \cite{BBCG} to illustrate the method. We expect our method to apply to a broad class of free fermionic solvable lattice models.

Throughout the paper, we identify the $+$ spin with $0$ and the $-$ spin with $1$. 

We start with the model in \cite{BBF}. Suppose that the rectangular lattice has $N$ rows and $\lambda_1+N$ columns. For each column with a given assignment of spins on the top and bottom boundary (the boundary condition encountered in this model is $\alpha\in\{0,1\}$ on the top and $0$ at the bottom), we view it as an operator called the ``column operator''. Specifically, for every $a\in\{1,\cdots,N\}$, we denote by $W_a$ a two-dimensional vector space over $\mathbb{C}$ spanned by the basis vectors $|0\rangle$ and $|1\rangle$. Then the column operator is an element of $End(W_1\otimes\cdots \otimes W_N)$. The precise definition will be given in Section \ref{Sect.3.2}. 

The partition function of the lattice model can be written as a certain component of the product of $\lambda_1+N$ column operators. The idea is to \underline{conjugate} the column operators so that the components of the conjugated operators have an explicit form. The key concepts involved in this conjugation are the ``permutation graph'' and the ``$F$-matrix'', which we will introduce in Sections \ref{Sect.3.3} and \ref{Sect.3.4}. The permutation graph is an $N$-site generalization of the $R$-matrix, which corresponds to the rotated vertex in the Yang-Baxter equation. It is an element of $End(W_1\otimes\cdots\otimes W_N)$ that depends on two permutations $\rho_1,\rho_2\in S_N$. The Yang-Baxter equation and another relation, the unitarity relation, are used to ensure that the permutation graph is well-defined. The $F$-matrix is then constructed based on the permutation graph. It is also an element of $End(W_1\otimes \cdots \otimes W_N)$, and is used to conjugate the column operators. The explicit form of the conjugated column operators are given in Proposition \ref{P3} below, based on which the partition function of the lattice model can be evaluated. 

For the model in \cite{BBCG}, in Section \ref{Sect.5}, we generalize the definitions of the column operator, the permutation graph, and the $F$-matrix to incorporate both $\Delta$ ice and $\Gamma$ ice. Another complication comes from the cap vertices on the right boundary. Suppose that the rectangular lattice has $N=2r$ rows. There are $r$ cap vertices in total, and we define a ``cap vector'' $K$ based on all the cap vertices. In view of the conjugation procedure, if we denote the $F$-matrix by $F$, we need to analyze the components of the vector $FK$. This is done with the help of the caduceus relation, an additional relation that is more specific to Cartan types B and C. The components of the corresponding operators are given in Proposition \ref{P5}. The evaluation of the partition function of this lattice model in Theorem \ref{Theorem2} gives a new proof of a conjecture made in \cite{BBCG} (the conjecture was first proved by Motegi et al. \cite{MSW} using the Izergin-Korepin technique; see below for details).

Previously, for Boltzmann weights that are related to the quantum group $U_q(\widehat{\mathfrak{sl}_2})$, the permutation graph and the $F$-matrix are introduced in \cite{MS}. These are used in \cite{MZ} to evaluate the partition functions of certain vertex models that are related to Hall-Littlewood polynomials. Those Boltzmann weights are quite different from the free fermionic Boltzmann weights considered in this paper. For free fermionic lattice models, the definition and application of the permutation graph and the $F$-matrix involve quite a bit more complications than those for $U_q(\widehat{\mathfrak{sl}_2})$ Boltzmann weights. 

In the following, we discuss previous approaches for the computation of partition functions of free fermionic solvable lattice models and compare them with the method introduced in this paper. 

One approach is based on the combinatorics of Gelfand-Tsetlin patterns (e.g. \cite{BBF}) and Proctor patterns (e.g. \cite{Iva}). Specifically, it relies on the Yang-Baxter equation (and two additional relations, the caduceus relation and the fish relation, for the model in \cite{Iva}) to establish symmetries of the partition function normalized by certain factors. Based on such symmetries, it can be shown that the normalized partition function is independent of a certain parameter of the model. Specializing this parameter to a particular value reduces the six-vertex model to a five-vertex model, and the combinatorics of Gelfand-Testlin patterns (or Proctor patterns) and the Weyl character formula are used to evaluate the partition function of the five-vertex model. This approach requires non-trivial combinatorial arguments tailored for each specific model, and does not lead to an evaluation of the partition function of the lattice model in \cite{BBCG}. Other combinatorial approaches include methods based on bijection (see e.g. \cite{HamelKing}) and methods based on lattice paths (see e.g. \cite{Okada}), which are also tailored for specific models. The method introduced in this paper is more general (for example, it applies to the lattice model in \cite{BBCG}) and less specific to the details of the models.

Another approach is based on the Izergin-Korepin technique. The Izergin-Korepin technique was introduced by Korepin \cite{Kor} and Izergin \cite{Ize}, and involves obtaining sufficiently many properties that are satisfied by a class of partition functions to uniquely identify them (these properties include initial conditions and recursive relations). In the context of free fermionic solvable lattice models, the Izergin-Korepin technique was applied in \cite{Mot2} to obtain a generalization of the Tokuyama formula for factorial Schur functions. This was later extended in \cite{Mot3,MSW} to elliptic Felderhof models and generalizations of the models in \cite{BBCG,Iva}. The partition function of \cite{BBCG} is computed through this approach in \cite{MSW} based on the generalized models. 

In applying the Izergin-Korepin technique, one needs to find the correct set of properties that can uniquely identify the partition function. This may require generalizations of the lattice model at hand and can be problem specific. Also, the Izergin-Korepin technique is a method for \underline{verifying} a given form of the partition function, and does not provide much insights on how the explicit form can be \underline{discovered}. In comparison, the method introduced here is general and leads to a \underline{direct} computation of the partition function of a particular lattice model.

\bigskip

The rest of this paper is organized as follows. In Section \ref{Sect.2}, we review the Boltzmann weights and the Yang-Baxter equations that are involved in the lattice models from \cite{BBF} and \cite{BBCG}. Two additional relations, the unitarity relation and the caduceus relation, are also discussed in this section. Then we introduce the column operator, the permutation graph, and the $F$-matrix for type A lattice models--which include the model in \cite{BBF}--in Section \ref{Sect.3}. Equipped with these tools, we apply the method outlined in Section \ref{Sect.1.2} to the lattice model in \cite{BBF} and obtain an explicit evaluation of its partition function in Section \ref{Sect.4}. Section \ref{Sect.5} extends the concepts in Section \ref{Sect.3} to type B and C lattice models, which include the model in \cite{BBCG}. These combined with the strategy described in Section \ref{Sect.1.2} lead to an explicit expression for the partition function of the model in \cite{BBCG}, giving a new proof of the conjecture made in \cite{BBCG}.

\subsection{Acknowledgement}
The author wishes to thank Daniel Bump for his encouragement and many helpful conversations.

\section{Boltzmann weights and Yang-Baxter equations}\label{Sect.2}

In Section \ref{Sect.2.1}, we review the Boltzmann weights of vertices that are involved in the lattice models from \cite{BBF} and \cite{BBCG}. Then we review the Yang-Baxter equations and the unitarity relation that are satisfied by these Boltzmann weights in Section \ref{Sect.2.2}. For the model from \cite{BBCG}, an additional relation, the caduceus relation, is also discussed.

There are three types of vertices that we consider: the ordinary vertex, the $R$-vertex, and the cap vertex. Ordinary vertices are vertices in the rectangular lattice on which the lattice model is based. There are two types of ordinary vertices called $\Delta$ ice and $\Gamma$ ice. Only $\Gamma$ ice is involved in the model from \cite{BBF}. $R$-vertices are auxiliary rotated vertices that appear in the Yang-Baxter equations. There are four types of $R$-vertices called $\Delta\Delta$ ice, $\Delta\Gamma$ ice, $\Gamma\Delta$ ice, and $\Gamma\Gamma$ ice. Only $\Gamma\Gamma$ ice is involved in the model from \cite{BBF}. Finally, cap vertices are U-turn vertices on the right boundary of the rectangular lattice that connect two adjacent rows of $\Delta$ ice and $\Gamma$ ice. They are only involved in the model from \cite{BBCG}.

\subsection{Boltzmann weights}\label{Sect.2.1}

In this subsection, we review the Boltzmann weights of the three types of vertices following \cite{BBF,BBCG}.

We start with the ordinary vertices. The Boltzmann weight of an ordinary vertex depends on the spins on the four edges that are adjacent to it and two parameters called the deformation parameter (denoted by $v\in \mathbb{C}$) and the spectral parameter (denoted by $z\in\mathbb{C}$). The deformation parameter is fixed for a given lattice model, while the spectral parameter can vary across different rows. The Boltzmann weights for $\Gamma$ ice and $\Delta$ ice are listed in Figures \ref{Figure2.1}-\ref{Figure2.2}. Here for $\Delta$ ice the signs of the spins on the horizontal edges are switched (from $+$ to $-$ and from $-$ to $+$) compared to those in \cite{BBF,BBCG} in order to simplify the presentation in later sections.

\begin{figure}[!h]
\[
\begin{array}{|c|c|c|c|c|c|}
\hline
   a_1 & a_2 & b_1 & b_2 & c_1 & c_2\\
\hline
\gammaice{+}{+}{+}{+} &
  \gammaice{-}{-}{-}{-} &
  \gammaice{+}{-}{+}{-} &
  \gammaice{-}{+}{-}{+} &
  \gammaice{-}{+}{+}{-} &
  \gammaice{+}{-}{-}{+}\\
\hline
   1 & z & -v & z & z(1-v) & 1\\
\hline\end{array}\]
\caption{Boltzmann weights for $\Gamma$ ice with deformation parameter $v$ and spectral parameter $z$}
\label{Figure2.1}
\end{figure}

\begin{figure}[!h]
\[
\begin{array}{|c|c|c|c|c|c|}
\hline
   a_1 & a_2  & b_1 &  b_2  &  c_1 & c_2\\
\hline
\gammaice{+}{+}{+}{+} &
  \gammaice{-}{-}{-}{-} &
  \gammaice{+}{-}{+}{-} &
  \gammaice{-}{+}{-}{+} &
  \gammaice{-}{+}{+}{-}  &
   \gammaice{+}{-}{-}{+}\\
\hline
   1 & -v z  & 1 &   z  &  z(1-v) & 1\\
\hline\end{array}\]
\caption{Boltzmann weights for $\Delta$ ice with deformation parameter $v$ and spectral parameter $z$}
\label{Figure2.2}
\end{figure}

Now we introduce the $R$-vertices. They are rotated vertices that appear in the Yang-Baxter equations (see Section \ref{Sect.2.2} for details). The Boltzmann weight of an $R$-vertex depends on the spins on the four edges that are adjacent to it, the deformation parameter $v\in\mathbb{C}$, and two spectral parameters $z,z'\in\mathbb{C}$. The Boltzmann weights for the four types of $R$-vertices ($\Gamma\Gamma$ ice, $\Gamma\Delta$ ice, $\Delta\Gamma$ ice, and $\Delta\Delta$ ice) are shown in Figures \ref{Figure2.3}-\ref{Figure2.6}. Here for $\Gamma\Delta$ ice, $\Delta\Gamma$ ice, and $\Delta\Delta$ ice we switch the signs of certain spins (from $+$ to $-$ and from $-$ to $+$) compared to those in \cite{BBF,BBCG} in accordance with the change for $\Delta$ ice. We also take a different normalization of the Boltzmann weights compared to that in \cite{BBF,BBCG} so that the Boltzmann weight of the $a_1$ pattern equals $1$ for the four types of $R$-vertices.  

\begin{figure}[!h]
\[\begin{array}{|c|c|c|c|c|c|}
\hline
   a_1 & a_2 & b_1 & b_2 & c_1 & c_2\\
\hline
\begin{tikzpicture}[scale=0.7]
\draw (0,0) to [out = 0, in = 180] (2,2);
\draw (0,2) to [out = 0, in = 180] (2,0);
\draw[fill=white] (0,0) circle (.35);
\draw[fill=white] (0,2) circle (.35);
\draw[fill=white] (2,0) circle (.35);
\draw[fill=white] (2,2) circle (.35);
\node at (0,0) {$+$};
\node at (0,2) {$+$};
\node at (2,2) {$+$};
\node at (2,0) {$+$};
\node at (2,0) {$+$};
\path[fill=white] (1,1) circle (.3);
\node at (1,1) {$R_{z,z'}$};
\end{tikzpicture}&
\begin{tikzpicture}[scale=0.7]
\draw (0,0) to [out = 0, in = 180] (2,2);
\draw (0,2) to [out = 0, in = 180] (2,0);
\draw[fill=white] (0,0) circle (.35);
\draw[fill=white] (0,2) circle (.35);
\draw[fill=white] (2,0) circle (.35);
\draw[fill=white] (2,2) circle (.35);
\node at (0,0) {$-$};
\node at (0,2) {$-$};
\node at (2,2) {$-$};
\node at (2,0) {$-$};
\path[fill=white] (1,1) circle (.3);
\node at (1,1) {$R_{z,z'}$};
\end{tikzpicture}&
\begin{tikzpicture}[scale=0.7]
\draw (0,0) to [out = 0, in = 180] (2,2);
\draw (0,2) to [out = 0, in = 180] (2,0);
\draw[fill=white] (0,0) circle (.35);
\draw[fill=white] (0,2) circle (.35);
\draw[fill=white] (2,0) circle (.35);
\draw[fill=white] (2,2) circle (.35);
\node at (0,0) {$+$};
\node at (0,2) {$-$};
\node at (2,2) {$+$};
\node at (2,0) {$-$};
\path[fill=white] (1,1) circle (.3);
\node at (1,1) {$R_{z,z'}$};
\end{tikzpicture}&
\begin{tikzpicture}[scale=0.7]
\draw (0,0) to [out = 0, in = 180] (2,2);
\draw (0,2) to [out = 0, in = 180] (2,0);
\draw[fill=white] (0,0) circle (.35);
\draw[fill=white] (0,2) circle (.35);
\draw[fill=white] (2,0) circle (.35);
\draw[fill=white] (2,2) circle (.35);
\node at (0,0) {$-$};
\node at (0,2) {$+$};
\node at (2,2) {$-$};
\node at (2,0) {$+$};
\path[fill=white] (1,1) circle (.3);
\node at (1,1) {$R_{z,z'}$};
\end{tikzpicture}&
\begin{tikzpicture}[scale=0.7]
\draw (0,0) to [out = 0, in = 180] (2,2);
\draw (0,2) to [out = 0, in = 180] (2,0);
\draw[fill=white] (0,0) circle (.35);
\draw[fill=white] (0,2) circle (.35);
\draw[fill=white] (2,0) circle (.35);
\draw[fill=white] (2,2) circle (.35);
\node at (0,0) {$-$};
\node at (0,2) {$+$};
\node at (2,2) {$+$};
\node at (2,0) {$-$};
\path[fill=white] (1,1) circle (.3);
\node at (1,1) {$R_{z,z'}$};
\end{tikzpicture}&
\begin{tikzpicture}[scale=0.7]
\draw (0,0) to [out = 0, in = 180] (2,2);
\draw (0,2) to [out = 0, in = 180] (2,0);
\draw[fill=white] (0,0) circle (.35);
\draw[fill=white] (0,2) circle (.35);
\draw[fill=white] (2,0) circle (.35);
\draw[fill=white] (2,2) circle (.35);
\node at (0,0) {$+$};
\node at (0,2) {$-$};
\node at (2,2) {$-$};
\node at (2,0) {$+$};
\path[fill=white] (1,1) circle (.3);
\node at (1,1) {$R_{z,z'}$};
\end{tikzpicture}\\
\hline
1&\frac{z-vz'}{z'-vz}
&\frac{v(z-z')}{z'-vz}
&\frac{z-z'}{z'-vz}
&\frac{(1-v)z}{z'-vz}
&\frac{(1-v)z'}{z'-vz}\\
\hline
\end{array}\]
\caption{Boltzmann weights for $\Gamma\Gamma$ ice with deformation parameter $v$ and spectral parameters $z,z'$}
\label{Figure2.3}
\end{figure}

\begin{figure}[!h]
\[\begin{array}{|c|c|c|c|c|c|}
\hline
   a_1 & a_2 & b_1 & b_2 & c_1 & c_2\\
\hline
\begin{tikzpicture}[scale=0.7]
\draw (0,0) to [out = 0, in = 180] (2,2);
\draw (0,2) to [out = 0, in = 180] (2,0);
\draw[fill=white] (0,0) circle (.35);
\draw[fill=white] (0,2) circle (.35);
\draw[fill=white] (2,0) circle (.35);
\draw[fill=white] (2,2) circle (.35);
\node at (0,0) {$+$};
\node at (0,2) {$+$};
\node at (2,2) {$+$};
\node at (2,0) {$+$};
\node at (2,0) {$+$};
\path[fill=white] (1,1) circle (.3);
\node at (1,1) {$R_{z,z'}$};
\end{tikzpicture}&
\begin{tikzpicture}[scale=0.7]
\draw (0,0) to [out = 0, in = 180] (2,2);
\draw (0,2) to [out = 0, in = 180] (2,0);
\draw[fill=white] (0,0) circle (.35);
\draw[fill=white] (0,2) circle (.35);
\draw[fill=white] (2,0) circle (.35);
\draw[fill=white] (2,2) circle (.35);
\node at (0,0) {$-$};
\node at (0,2) {$-$};
\node at (2,2) {$-$};
\node at (2,0) {$-$};
\path[fill=white] (1,1) circle (.3);
\node at (1,1) {$R_{z,z'}$};
\end{tikzpicture}&
\begin{tikzpicture}[scale=0.7]
\draw (0,0) to [out = 0, in = 180] (2,2);
\draw (0,2) to [out = 0, in = 180] (2,0);
\draw[fill=white] (0,0) circle (.35);
\draw[fill=white] (0,2) circle (.35);
\draw[fill=white] (2,0) circle (.35);
\draw[fill=white] (2,2) circle (.35);
\node at (0,0) {$+$};
\node at (0,2) {$-$};
\node at (2,2) {$+$};
\node at (2,0) {$-$};
\path[fill=white] (1,1) circle (.3);
\node at (1,1) {$R_{z,z'}$};
\end{tikzpicture}&
\begin{tikzpicture}[scale=0.7]
\draw (0,0) to [out = 0, in = 180] (2,2);
\draw (0,2) to [out = 0, in = 180] (2,0);
\draw[fill=white] (0,0) circle (.35);
\draw[fill=white] (0,2) circle (.35);
\draw[fill=white] (2,0) circle (.35);
\draw[fill=white] (2,2) circle (.35);
\node at (0,0) {$-$};
\node at (0,2) {$+$};
\node at (2,2) {$-$};
\node at (2,0) {$+$};
\path[fill=white] (1,1) circle (.3);
\node at (1,1) {$R_{z,z'}$};
\end{tikzpicture}&
\begin{tikzpicture}[scale=0.7]
\draw (0,0) to [out = 0, in = 180] (2,2);
\draw (0,2) to [out = 0, in = 180] (2,0);
\draw[fill=white] (0,0) circle (.35);
\draw[fill=white] (0,2) circle (.35);
\draw[fill=white] (2,0) circle (.35);
\draw[fill=white] (2,2) circle (.35);
\node at (0,0) {$-$};
\node at (0,2) {$+$};
\node at (2,2) {$+$};
\node at (2,0) {$-$};
\path[fill=white] (1,1) circle (.3);
\node at (1,1) {$R_{z,z'}$};
\end{tikzpicture}&
\begin{tikzpicture}[scale=0.7]
\draw (0,0) to [out = 0, in = 180] (2,2);
\draw (0,2) to [out = 0, in = 180] (2,0);
\draw[fill=white] (0,0) circle (.35);
\draw[fill=white] (0,2) circle (.35);
\draw[fill=white] (2,0) circle (.35);
\draw[fill=white] (2,2) circle (.35);
\node at (0,0) {$+$};
\node at (0,2) {$-$};
\node at (2,2) {$-$};
\node at (2,0) {$+$};
\path[fill=white] (1,1) circle (.3);
\node at (1,1) {$R_{z,z'}$};
\end{tikzpicture}\\
\hline
1&\frac{z'-vz}{z-vz'}
&\frac{v(z-z')}{z-vz'}
&\frac{z-z'}{z-vz'}
&\frac{(1-v)z}{z-vz'}
&\frac{(1-v)z'}{z-vz'}\\
\hline
\end{array}\]
\caption{Boltzmann weights for $\Delta\Delta$ ice with deformation parameter $v$ and spectral parameters $z,z'$}
\label{Figure2.4}
\end{figure}

\begin{figure}[!h]
\[\begin{array}{|c|c|c|c|c|c|}
\hline
   a_1 & a_2 & b_1 & b_2 & c_1 & c_2\\
\hline
\begin{tikzpicture}[scale=0.7]
\draw (0,0) to [out = 0, in = 180] (2,2);
\draw (0,2) to [out = 0, in = 180] (2,0);
\draw[fill=white] (0,0) circle (.35);
\draw[fill=white] (0,2) circle (.35);
\draw[fill=white] (2,0) circle (.35);
\draw[fill=white] (2,2) circle (.35);
\node at (0,0) {$+$};
\node at (0,2) {$+$};
\node at (2,2) {$+$};
\node at (2,0) {$+$};
\node at (2,0) {$+$};
\path[fill=white] (1,1) circle (.3);
\node at (1,1) {$R_{z,z'}$};
\end{tikzpicture}&
\begin{tikzpicture}[scale=0.7]
\draw (0,0) to [out = 0, in = 180] (2,2);
\draw (0,2) to [out = 0, in = 180] (2,0);
\draw[fill=white] (0,0) circle (.35);
\draw[fill=white] (0,2) circle (.35);
\draw[fill=white] (2,0) circle (.35);
\draw[fill=white] (2,2) circle (.35);
\node at (0,0) {$-$};
\node at (0,2) {$-$};
\node at (2,2) {$-$};
\node at (2,0) {$-$};
\path[fill=white] (1,1) circle (.3);
\node at (1,1) {$R_{z,z'}$};
\end{tikzpicture}&
\begin{tikzpicture}[scale=0.7]
\draw (0,0) to [out = 0, in = 180] (2,2);
\draw (0,2) to [out = 0, in = 180] (2,0);
\draw[fill=white] (0,0) circle (.35);
\draw[fill=white] (0,2) circle (.35);
\draw[fill=white] (2,0) circle (.35);
\draw[fill=white] (2,2) circle (.35);
\node at (0,0) {$+$};
\node at (0,2) {$-$};
\node at (2,2) {$+$};
\node at (2,0) {$-$};
\path[fill=white] (1,1) circle (.3);
\node at (1,1) {$R_{z,z'}$};
\end{tikzpicture}&
\begin{tikzpicture}[scale=0.7]
\draw (0,0) to [out = 0, in = 180] (2,2);
\draw (0,2) to [out = 0, in = 180] (2,0);
\draw[fill=white] (0,0) circle (.35);
\draw[fill=white] (0,2) circle (.35);
\draw[fill=white] (2,0) circle (.35);
\draw[fill=white] (2,2) circle (.35);
\node at (0,0) {$-$};
\node at (0,2) {$+$};
\node at (2,2) {$-$};
\node at (2,0) {$+$};
\path[fill=white] (1,1) circle (.3);
\node at (1,1) {$R_{z,z'}$};
\end{tikzpicture}&
\begin{tikzpicture}[scale=0.7]
\draw (0,0) to [out = 0, in = 180] (2,2);
\draw (0,2) to [out = 0, in = 180] (2,0);
\draw[fill=white] (0,0) circle (.35);
\draw[fill=white] (0,2) circle (.35);
\draw[fill=white] (2,0) circle (.35);
\draw[fill=white] (2,2) circle (.35);
\node at (0,0) {$-$};
\node at (0,2) {$+$};
\node at (2,2) {$+$};
\node at (2,0) {$-$};
\path[fill=white] (1,1) circle (.3);
\node at (1,1) {$R_{z,z'}$};
\end{tikzpicture}&
\begin{tikzpicture}[scale=0.7]
\draw (0,0) to [out = 0, in = 180] (2,2);
\draw (0,2) to [out = 0, in = 180] (2,0);
\draw[fill=white] (0,0) circle (.35);
\draw[fill=white] (0,2) circle (.35);
\draw[fill=white] (2,0) circle (.35);
\draw[fill=white] (2,2) circle (.35);
\node at (0,0) {$+$};
\node at (0,2) {$-$};
\node at (2,2) {$-$};
\node at (2,0) {$+$};
\path[fill=white] (1,1) circle (.3);
\node at (1,1) {$R_{z,z'}$};
\end{tikzpicture}\\
\hline
1&1
&\frac{z'-v^2z}{z'-vz}
&\frac{z-z'}{z'-vz}
&\frac{(1-v)z}{z'-vz}
&\frac{(1-v)z'}{z'-vz}\\
\hline
\end{array}\]
\caption{Boltzmann weights for $\Delta\Gamma$ ice with deformation parameter $v$ and spectral parameters $z,z'$}
\label{Figure2.5}
\end{figure}

\begin{figure}[!h]
\[\begin{array}{|c|c|c|c|c|c|}
\hline
   a_1 & a_2 & b_1 & b_2 & c_1 & c_2\\
\hline
\begin{tikzpicture}[scale=0.7]
\draw (0,0) to [out = 0, in = 180] (2,2);
\draw (0,2) to [out = 0, in = 180] (2,0);
\draw[fill=white] (0,0) circle (.35);
\draw[fill=white] (0,2) circle (.35);
\draw[fill=white] (2,0) circle (.35);
\draw[fill=white] (2,2) circle (.35);
\node at (0,0) {$+$};
\node at (0,2) {$+$};
\node at (2,2) {$+$};
\node at (2,0) {$+$};
\node at (2,0) {$+$};
\path[fill=white] (1,1) circle (.3);
\node at (1,1) {$R_{z,z'}$};
\end{tikzpicture}&
\begin{tikzpicture}[scale=0.7]
\draw (0,0) to [out = 0, in = 180] (2,2);
\draw (0,2) to [out = 0, in = 180] (2,0);
\draw[fill=white] (0,0) circle (.35);
\draw[fill=white] (0,2) circle (.35);
\draw[fill=white] (2,0) circle (.35);
\draw[fill=white] (2,2) circle (.35);
\node at (0,0) {$-$};
\node at (0,2) {$-$};
\node at (2,2) {$-$};
\node at (2,0) {$-$};
\path[fill=white] (1,1) circle (.3);
\node at (1,1) {$R_{z,z'}$};
\end{tikzpicture}&
\begin{tikzpicture}[scale=0.7]
\draw (0,0) to [out = 0, in = 180] (2,2);
\draw (0,2) to [out = 0, in = 180] (2,0);
\draw[fill=white] (0,0) circle (.35);
\draw[fill=white] (0,2) circle (.35);
\draw[fill=white] (2,0) circle (.35);
\draw[fill=white] (2,2) circle (.35);
\node at (0,0) {$+$};
\node at (0,2) {$-$};
\node at (2,2) {$+$};
\node at (2,0) {$-$};
\path[fill=white] (1,1) circle (.3);
\node at (1,1) {$R_{z,z'}$};
\end{tikzpicture}&
\begin{tikzpicture}[scale=0.7]
\draw (0,0) to [out = 0, in = 180] (2,2);
\draw (0,2) to [out = 0, in = 180] (2,0);
\draw[fill=white] (0,0) circle (.35);
\draw[fill=white] (0,2) circle (.35);
\draw[fill=white] (2,0) circle (.35);
\draw[fill=white] (2,2) circle (.35);
\node at (0,0) {$-$};
\node at (0,2) {$+$};
\node at (2,2) {$-$};
\node at (2,0) {$+$};
\path[fill=white] (1,1) circle (.3);
\node at (1,1) {$R_{z,z'}$};
\end{tikzpicture}&
\begin{tikzpicture}[scale=0.7]
\draw (0,0) to [out = 0, in = 180] (2,2);
\draw (0,2) to [out = 0, in = 180] (2,0);
\draw[fill=white] (0,0) circle (.35);
\draw[fill=white] (0,2) circle (.35);
\draw[fill=white] (2,0) circle (.35);
\draw[fill=white] (2,2) circle (.35);
\node at (0,0) {$-$};
\node at (0,2) {$+$};
\node at (2,2) {$+$};
\node at (2,0) {$-$};
\path[fill=white] (1,1) circle (.3);
\node at (1,1) {$R_{z,z'}$};
\end{tikzpicture}&
\begin{tikzpicture}[scale=0.7]
\draw (0,0) to [out = 0, in = 180] (2,2);
\draw (0,2) to [out = 0, in = 180] (2,0);
\draw[fill=white] (0,0) circle (.35);
\draw[fill=white] (0,2) circle (.35);
\draw[fill=white] (2,0) circle (.35);
\draw[fill=white] (2,2) circle (.35);
\node at (0,0) {$+$};
\node at (0,2) {$-$};
\node at (2,2) {$-$};
\node at (2,0) {$+$};
\path[fill=white] (1,1) circle (.3);
\node at (1,1) {$R_{z,z'}$};
\end{tikzpicture}\\
\hline
1&1
&\frac{v^2z'-z}{z-vz'}
&\frac{z-z'}{z-vz'}
&\frac{(1-v)z}{z-vz'}
&\frac{(1-v)z'}{z-vz'}\\
\hline
\end{array}\]
\caption{Boltzmann weights for $\Gamma\Delta$ ice with deformation parameter $v$ and spectral parameters $z,z'$}
\label{Figure2.6}
\end{figure}

Finally we introduce the cap vertices. These are used in the lattice model from \cite{BBCG}. The Boltzmann weight of a cap vertex depends on the spins on the two edges that are adjacent to it, the deformation parameter $v\in\mathbb{C}$, and the spectral parameter $z\in\mathbb{C}$. The Boltzmann weights are shown in Figure \ref{Figure2.7}. Here, the sign of the spin on the top edge is switched (from $+$ to $-$ and from $-$ to $+$) in accordance with the change for $\Delta$ ice. We note that the model from \cite{Iva} involves a different set of Boltzmann weights for the cap vertices. Our method applies to that model as well.

\begin{figure}[!h]
\[
\begin{array}{|c|c|c|c|c|c|}
\hline
\text{Cap} &\caps{-}{+} & \caps{+}{-} \\
\hline
\text{Boltzmann weight}  & -\sqrt{v} z^{1\slash 2}  &  z^{-1\slash 2} \\
\hline\end{array}\]
\caption{Boltzmann weights for a cap vertex with deformation parameter $v$ and spectral parameter $z$}
\label{Figure2.7}
\end{figure}

\subsection{Yang-Baxter equations, unitarity relation, and caduceus relation}\label{Sect.2.2}

In this subsection, we review several relations that are satisfied by the vertices introduced in Section \ref{Sect.2.1}. These include two sets of Yang-Baxter equations known as the ``$RTT$ relation'' and the ``$RRR$ relation'', the unitarity relation, and the caduceus relation.

We start with the Yang-Baxter equations. Two ordinary vertices and an $R$-vertex satisfy the following set of Yang-Baxter equations known as the ``$RTT$ relation''. These relations were obtained in \cite[Theorem 9]{BBF2} (\cite{BBF2} is the arXiv version of \cite{BBF}).

\begin{proposition}[\cite{BBF2}, Theorem 8]\label{YBE1}
For any $X,Y\in \{\Gamma,\Delta\}$ the following holds. Assume that $S$ is $X$ ice with spectral parameter $z_i$, $T$ is $Y$ ice with spectral parameter $z_j$, and $R$ is $XY$ ice with spectral parameters $z_i,z_j$. Then the partition functions of the following two configurations are equal for any fixed combination of spins $a,b,c,d,e,f$.
\begin{equation}
\hfill
\begin{tikzpicture}[baseline=(current bounding box.center)]
  \draw (0,1) to [out = 0, in = 180] (2,3) to (4,3);
  \draw (0,3) to [out = 0, in = 180] (2,1) to (4,1);
  \draw (3,0) to (3,4);
  \draw[fill=white] (0,1) circle (.3);
  \draw[fill=white] (0,3) circle (.3);
  \draw[fill=white] (3,4) circle (.3);
  \draw[fill=white] (4,3) circle (.3);
  \draw[fill=white] (4,1) circle (.3);
  \draw[fill=white] (3,0) circle (.3);
  \draw[fill=white] (2,3) circle (.3);
  \draw[fill=white] (2,1) circle (.3);
  \draw[fill=white] (3,2) circle (.3);
  \node at (0,1) {$a$};
  \node at (0,3) {$b$};
  \node at (3,4) {$c$};
  \node at (4,3) {$d$};
  \node at (4,1) {$e$};
  \node at (3,0) {$f$};
  \node at (2,3) {$g$};
  \node at (3,2) {$h$};
  \node at (2,1) {$i$};
\filldraw[black] (3,3) circle (2pt);
\node at (3,3) [anchor=south west] {$S$};
\filldraw[black] (3,1) circle (2pt);
\node at (3,1) [anchor=north west] {$T$};
\filldraw[black] (1,2) circle (2pt);
\node at (1,2) [anchor=west] {$R$};
\end{tikzpicture}\qquad\qquad
\begin{tikzpicture}[baseline=(current bounding box.center)]
  \draw (0,1) to (2,1) to [out = 0, in = 180] (4,3);
  \draw (0,3) to (2,3) to [out = 0, in = 180] (4,1);
  \draw (1,0) to (1,4);
  \draw[fill=white] (0,1) circle (.3);
  \draw[fill=white] (0,3) circle (.3);
  \draw[fill=white] (1,4) circle (.3);
  \draw[fill=white] (4,3) circle (.3);
  \draw[fill=white] (4,1) circle (.3);
  \draw[fill=white] (1,0) circle (.3);
  \draw[fill=white] (2,3) circle (.3);
  \draw[fill=white] (1,2) circle (.3);
  \draw[fill=white] (2,1) circle (.3);
  \node at (0,1) {$a$};
  \node at (0,3) {$b$};
  \node at (1,4) {$c$};
  \node at (4,3) {$d$};
  \node at (4,1) {$e$};
  \node at (1,0) {$f$};
  \node at (2,3) {$j$};
  \node at (1,2) {$k$};
  \node at (2,1) {$l$};
\filldraw[black] (1,3) circle (2pt);
\node at (1,3) [anchor=south west] {$T$};
\filldraw[black] (1,1) circle (2pt);
\node at (1,1) [anchor=north west]{$S$};
\filldraw[black] (3,2) circle (2pt);
\node at (3,2) [anchor=west] {$R$};
\end{tikzpicture}
\end{equation}
\end{proposition}

Three $R$-vertices satisfy another set of Yang-Baxter equations known as the ``$RRR$ relation''. These relations were obtained in \cite[Theorem 10]{BBF2}.

\begin{proposition}[\cite{BBF2}, Theorem 10]\label{YBE2}
For any $X,Y,Z\in\{\Gamma,\Delta\}$ the following holds. Assume that $R$ is $XY$ ice with spectral parameters $z_i,z_j$, $S$ is $XZ$ ice with spectral parameters $z_i,z_k$, and $T$ is $YZ$ ice with spectral parameters $z_j,z_k$. Then the partition functions of the following two configurations are equal for any fixed combination of spins $a,b,c,d,e,f$.
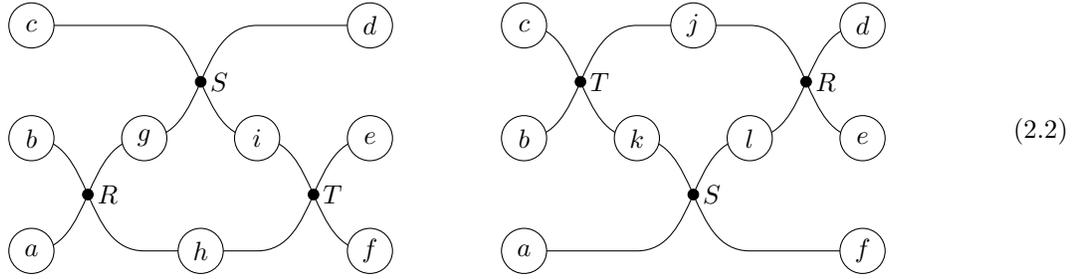
\begin{equation}
\hfill
\begin{tikzpicture}[baseline=(current bounding box.center)]
  \draw (0,0) to [out = 0, in = 180] (1.5,1.5) to [out = 0, in = 180] (3,3) to (4.5,3);
  \draw (0,3) to (1.5,3) to [out = 0, in = 180] (3,1.5) to [out=0,in=180] (4.5,0);
  \draw (0,1.5) to [out=0,in=180] (1.5,0) to (3,0) to [out=0, in=180] (4.5,1.5);
  \draw[fill=white] (0,0) circle (.3);
  \draw[fill=white] (0,1.5) circle (.3);
  \draw[fill=white] (0,3) circle (.3);
  \draw[fill=white] (2.25,0) circle (.3);
  \draw[fill=white] (1.5,1.5) circle (.3);
  \draw[fill=white] (3,1.5) circle (.3);
  \draw[fill=white] (4.5,0) circle (.3);
  \draw[fill=white] (4.5,1.5) circle (.3);
  \draw[fill=white] (4.5,3) circle (.3);
  \node at (0,0) {$a$};
  \node at (0,1.5) {$b$};
  \node at (0,3) {$c$};
  \node at (4.5,3) {$d$};
  \node at (4.5,1.5) {$e$};
  \node at (4.5,0) {$f$};
  \node at (1.5,1.5) {$g$};
  \node at (2.25,0) {$h$};
  \node at (3,1.5) {$i$};
\filldraw[black] (0.75,0.75) circle (2pt);
\node at (0.75,0.75) [anchor=west] {$R$};
\filldraw[black] (2.25,2.25) circle (2pt);
\node at (2.25,2.25) [anchor=west] {$S$};
\filldraw[black] (3.75,0.75) circle (2pt);
\node at (3.75,0.75) [anchor=west] {$T$};
\end{tikzpicture}\qquad\qquad
\begin{tikzpicture}[baseline=(current bounding box.center)]
  \draw (0,0) to (1.5,0) to [out = 0, in = 180] (3,1.5) to [out=0, in=180] (4.5,3);
  \draw (0,3) to [out=0,in=180] (1.5,1.5) to [out = 0, in = 180] (3,0) to (4.5,0);
  \draw (0,1.5) to [out=0,in=180] (1.5,3)  to (3,3) to [out=0, in=180] (4.5,1.5);
  \draw[fill=white] (0,0) circle (.3);
  \draw[fill=white] (0,1.5) circle (.3);
  \draw[fill=white] (0,3) circle (.3);
  \draw[fill=white] (4.5,0) circle (.3);
  \draw[fill=white] (4.5,1.5) circle (.3);
  \draw[fill=white] (4.5,3) circle (.3);
  \draw[fill=white] (2.25,3) circle (.3);
  \draw[fill=white] (1.5,1.5) circle (.3);
  \draw[fill=white] (3,1.5) circle (.3);
  \node at (0,0) {$a$};
  \node at (0,1.5) {$b$};
  \node at (0,3) {$c$};
  \node at (4.5,3) {$d$};
  \node at (4.5,1.5) {$e$};
  \node at (4.5,0) {$f$};
  \node at (2.25,3) {$j$};
  \node at (1.5,1.5) {$k$};
  \node at (3,1.5) {$l$};
\filldraw[black] (0.75,2.25) circle (2pt);
\node at (0.75,2.25) [anchor=west] {$T$};
\filldraw[black] (2.25,0.75) circle (2pt);
\node at (2.25,0.75) [anchor=west]{$S$};
\filldraw[black] (3.75,2.25) circle (2pt);
\node at (3.75,2.25) [anchor=west] {$R$};
\end{tikzpicture}
\end{equation}
\end{proposition}

The $R$-vertices also satisfy the unitarity relation as given by the following theorem. We note again that the normalization of the $R$-vertices here is different from that of \cite{BBF,BBCG}.

\begin{proposition}\label{Unit}
For any $X,Y\in\{\Gamma,\Delta\}$ the following holds. Assume that $S$ is $XY$ ice with spectral parameters $z_i,z_j$, and $T$ is $YX$ ice with spectral parameters $z_j,z_i$. Then the partition function of the following two configurations are equal for any fixed combination of spins $a,b,c,d$. Here, the partition function of the right configuration is $\mathbbm{1}_{a=d} \mathbbm{1}_{b=c}$.
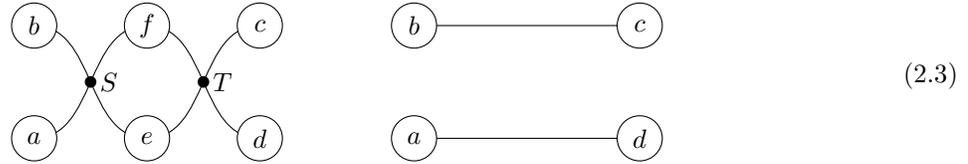
\begin{equation}
\hfill
\begin{tikzpicture}[baseline=(current bounding box.center)]
  \draw (0,0) to [out = 0, in = 180] (1.5,1.5) to [out = 0, in = 180] (3,0);
  \draw (0,1.5) to [out=0,in=180] (1.5,0) to [out=0,in=180] (3,1.5);
  \draw[fill=white] (0,0) circle (.3);
  \draw[fill=white] (0,1.5) circle (.3);
  \draw[fill=white] (1.5,0) circle (.3);
  \draw[fill=white] (1.5,1.5) circle (.3);
  \draw[fill=white] (3,0) circle (.3);
  \draw[fill=white] (3,1.5) circle (.3);
  \node at (0,0) {$a$};
  \node at (0,1.5) {$b$};
  \node at (1.5,0) {$e$};
  \node at (1.5,1.5) {$f$};
  \node at (3,0) {$d$};
  \node at (3,1.5) {$c$};
\filldraw[black] (0.75,0.75) circle (2pt);
\node at (0.75,0.75) [anchor=west] {$S$};
\filldraw[black] (2.25,0.75) circle (2pt);
\node at (2.25,0.75) [anchor=west] {$T$};
\end{tikzpicture}\qquad\qquad
\begin{tikzpicture}[baseline=(current bounding box.center)]
 \draw (0,0) to (3,0);
  \draw (0,1.5) to (3,1.5);
  \draw[fill=white] (0,0) circle (.3);
  \draw[fill=white] (0,1.5) circle (.3);
  \draw[fill=white] (3,0) circle (.3);
  \draw[fill=white] (3,1.5) circle (.3);
  \node at (0,0) {$a$};
  \node at (0,1.5) {$b$};
  \node at (3,0) {$d$};
  \node at (3,1.5) {$c$};
\end{tikzpicture}
\end{equation}
\end{proposition}
\begin{proof}
The relation is checked using a SAGE program.
\end{proof}

The four types of $R$-vertices and the cap vertices also satisfy the following relation called the ``caduceus relation''. It is used in the lattice model from \cite{BBCG}. Note that the normalization of the $R$-vertices here is different from that in \cite{BBCG}.

\begin{proposition}[\cite{BBCG}]\label{caduc}
Assume that $A$ is $\Delta\Delta$ ice of spectral parameters $z_i,z_j$, $B$ is $\Gamma\Gamma$ ice of spectral parameters $z_i^{-1},z_j^{-1}$, $C$ is $\Gamma\Delta$ ice of spectral parameters $z_i^{-1},z_j$, and $D$ is $\Delta\Gamma$ ice of spectral parameters $z_i,z_j^{-1}$. Also assume that the cap vertices $K_1,K_2$ have spectral parameters $z_i$ and $z_j$, respectively. Denote by $Z(I_1(\epsilon_1,\epsilon_2,\epsilon_3,\epsilon_4))$ the partition function of the following configuration with fixed combination of spins $\epsilon_1,\epsilon_2,\epsilon_3,\epsilon_4$.
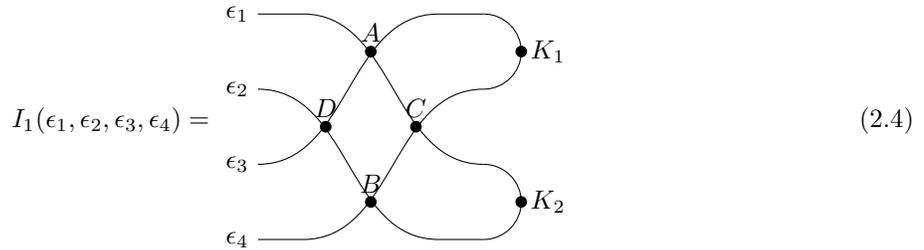
\begin{equation}
\label{eqn:caduceus1}
\hfill
I_1(\epsilon_1,\epsilon_2,\epsilon_3,\epsilon_4)=
\begin{tikzpicture}[baseline=(current bounding box.center)]
  \draw (0,0) to (0.6,0) to [out=0, in=180] (3,2);
  \draw (0,3) to (0.6,3) to [out=0, in=180] (3,1);
  \draw (0,1) to [out=0, in=180] (2.4,3) to (3,3) ;
  \draw (0,2) to [out=0, in=180] (2.4,0) to (3,0);
  \draw (3,2) arc(-90:90:0.5);
  \draw (3,0) arc(-90:90:0.5);
  \filldraw[black] (3.5,0.5) circle (2pt);
  \filldraw[black] (3.5,2.5) circle (2pt);
  \filldraw[black] (0.9,1.5) circle (2pt);
  \filldraw[black] (2.1,1.5) circle (2pt);
  \filldraw[black] (1.5,0.5) circle (2pt);
  \filldraw[black] (1.5,2.5) circle (2pt);
  \node at (0.9,1.5) [anchor=south] {$D$};
  \node at (2.1,1.5) [anchor=south] {$C$};
  \node at (1.5,0.5) [anchor=south] {$B$};
  \node at (1.5,2.5) [anchor=south] {$A$};
  \node at (0,0) [anchor=east] {$\epsilon_4$};
  \node at (0,1) [anchor=east] {$\epsilon_3$};
  \node at (0,3) [anchor=east] {$\epsilon_1$};
  \node at (0,2) [anchor=east] {$\epsilon_2$};
  \node at (3.5,2.5) [anchor=west] {$K_1$};
  \node at (3.5,0.5) [anchor=west] {$K_2$};
\end{tikzpicture}
\end{equation}
Also denote by $Z(I_2(\epsilon_1,\epsilon_2,\epsilon_3,\epsilon_4))$ the partition function of the following configuration with fixed combination of spins $\epsilon_1,\epsilon_2,\epsilon_3,\epsilon_4$.
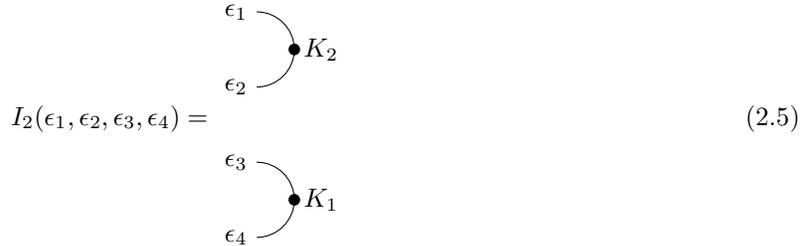
\begin{equation}
\hfill
I_2(\epsilon_1,\epsilon_2,\epsilon_3,\epsilon_4)=
\begin{tikzpicture}[baseline=(current bounding box.center)]
  \draw (0,2) arc(-90:90:0.5);
  \draw (0,0) arc(-90:90:0.5);
  \filldraw[black] (0.5,0.5) circle (2pt);
  \filldraw[black] (0.5,2.5) circle (2pt);
  \node at (0,0) [anchor=east] {$\epsilon_4$};
  \node at (0,1) [anchor=east] {$\epsilon_3$};
  \node at (0,3) [anchor=east] {$\epsilon_1$};
  \node at (0,2) [anchor=east] {$\epsilon_2$};
  \node at (0.5,2.5) [anchor=west] {$K_2$};
  \node at (0.5,0.5) [anchor=west] {$K_1$};
\end{tikzpicture}
\end{equation}
Then for any fixed combination of spins $\epsilon_1,\epsilon_2,\epsilon_3,\epsilon_4$, we have
\begin{equation}
    Z(I_1(\epsilon_1,\epsilon_2,\epsilon_3,\epsilon_4))=\frac{z_j-vz_i}{z_i-vz_j}Z(I_2(\epsilon_1,\epsilon_2,\epsilon_3,\epsilon_4)).
\end{equation}
\end{proposition}

\section{Column operator, permutation graph, and $F$-matrix}\label{Sect.3}

In this section, we introduce three key concepts that are used in our method for computing the partition function: the column operator, the permutation graph, and the $F$-matrix. We focus on type A lattice models in this section, and defer the generalization to type B and C lattice models to Section \ref{Sect.5}. Therefore, in this section, only $\Gamma$ ice is involved for ordinary vertices, and only $\Gamma\Gamma$ ice is involved for $R$-vertices.

Some basic notations are given in Section \ref{Sect.3.1}. Then we introduce the three concepts in Sections \ref{Sect.3.2}-\ref{Sect.3.4}, respectively. In Section \ref{Sect.3.5}, we derive some basic properties of the $F$-matrix.

\subsection{Basic notations}\label{Sect.3.1}
In this subsection, we set up some basic notations. For any $a\in\{1,2,\cdots\}$, we let $W_a \cong \mathbb{C}^2$ be a 2-dimensional vector space over $\mathbb{C}$ spanned by two basis vectors $|0\rangle$ and $|1\rangle$ (here, as mentioned before, $0$ corresponds to the $+$ spin and $1$ corresponds to the $-$ spin). For any $i,j\in\{0,1\}$, we denote by $E_a^{(i,j)}$ the $2\times 2$ elementary matrix acting on $W_a$ with $1$ at position $(i,j)$ and $0$ elsewhere (the rows and columns of the matrix are labeled by $0,1$). 

Now we discuss the $R$-matrix. For any $a,b,c,d\in\{0,1\}$ and any $x_i,x_j\in\mathbb{C}$, we denote by $R(a,b,c,d;x_i,x_j)$ the Boltzmann weight of the following $R$-vertex with spectral parameters $x_i,x_j$:
\begin{equation}
\begin{tikzpicture}[scale=0.7]
\draw (0,0) to [out = 0, in = 180] (2,2);
\draw (0,2) to [out = 0, in = 180] (2,0);
\draw[fill=white] (0,0) circle (.35);
\draw[fill=white] (0,2) circle (.35);
\draw[fill=white] (2,0) circle (.35);
\draw[fill=white] (2,2) circle (.35);
\node at (0,0) {$a$};
\node at (0,2) {$b$};
\node at (2,2) {$c$};
\node at (2,0) {$d$};
\path[fill=white] (1,1) circle (.3);
\node at (1,1) {$R_{x_i,x_j}$};
\end{tikzpicture}
\end{equation}
For any two distinct positive integers $i,j$, we define the $R$-matrix $R_{i,j}(x_i,x_j)$ with spectral parameters  $x_i,x_j$ that acts on $W_i \otimes W_j$ as follows:
\begin{equation}
    R_{i,j}(x_i,x_j)=\sum_{a,b,c,d\in\{0,1\}}R(a,b,c,d;x_i,x_j) E_i^{(a,c)} E_j^{(b,d)}.
\end{equation}
We also denote $R(x_1,x_2):=R_{12}(x_1,x_2)$.

We also use the following notations for the Boltzmann weights. For $\Gamma$ ice with spectral parameter $x_i$, we denote by  $a_1(x_i)$ the Boltzmann weight of the $a_1$ state (see Figure \ref{Figure2.1}), and similarly for the other states. For $\Gamma\Gamma$ ice with spectral parameters $x_i,x_j$, we denote by $a_1(x_i,x_j)$ the Boltzmann weight of the $a_1$ state (see Figure \ref{Figure2.3}), and similarly for the other states.

In the following, we fix a positive integer $N$. For any $(i_1,i_2,\cdots,i_N)\in \{0,1\}^N$, we denote by $|i_1,i_2,\cdots,i_N\rangle=|i_1\rangle\otimes  |i_2\rangle \otimes \cdots \otimes |i_N\rangle$ the corresponding basis vector of $W_1 \otimes \cdots \otimes W_N$, and $\langle i_1,i_2,\cdots,i_N|\in (W_1\otimes \cdots  \otimes W_N)^{*}$ the dual vector of $|i_1,i_2,\cdots,i_N\rangle$. Then for any operator $A\in End(W_1\otimes \cdots \otimes W_N)$, we define the component $(A)_{i_1\cdots i_N}^{j_1\cdots j_N}$ by
\begin{equation}\label{Oped}
    A|j_1,\cdots,j_N\rangle =\sum_{(i_1,\cdots,i_N)\in\{0,1\}^{N}}(A)_{i_1\cdots i_N}^{j_1\cdots j_N}|i_1,\cdots,i_N\rangle,\text{ for any }(j_1,\cdots,j_N)\in\{0,1\}^N.
\end{equation}
Similarly, for any $a\in W_1\otimes \dots \otimes W_N$, the component $(a)_{i_1,\cdots,i_N}$ is defined by
\begin{equation}
    (a)_{i_1,\cdots,i_N}=\langle i_1,\cdots,i_N|a, \text{ for any }(i_1,\cdots,i_N)\in\{0,1\}^N;
\end{equation}
for any $a\in (W_1\otimes \cdots  \otimes W_N)^{*}$, the component $(a)^{i_1\cdots i_N}$ is defined by
\begin{equation}
    (a)^{i_1,\cdots,i_N}=a |i_1,\cdots,i_N\rangle, \text{ for any }(i_1,\cdots,i_N)\in\{0,1\}^N.
\end{equation}

\subsection{Column operator}\label{Sect.3.2}

In this subsection, we introduce the concept of ``column operator''. Namely, for any $\alpha\in\{0,1\}$ and $\vec{x}=(x_1,\cdots,x_N)\in\mathbb{C}^N$, we define the column operator $S^{[\alpha]}(\vec{x})\in End(W_1\otimes \cdots \otimes W_N)$ by specifying its components $(S^{[\alpha]}(\vec{x}))_{i_1\cdots i_N}^{j_1\cdots j_N}$ for any $(i_1,\cdots,i_N),(j_1,\cdots,j_N)\in\{0,1\}^N$.

To define $(S^{[\alpha]}(\vec{x}))_{i_1\cdots i_N}^{j_1\cdots j_N}$, recall that we have identified $+$ spin with $0$ and $-$ spin with $1$. Consider a column of ordinary vertices whose spectral parameters are given by $x_1,\cdots,x_N$ from bottom to top. We also specify the boundary condition as follows: the top edge is labeled $\alpha$, the bottom edge is labeled $0$, the left edges are labeled $i_1,\cdots,i_N$ (from bottom to top), and the right edges are labeled $j_1,\cdots,j_N$ (from bottom to top). The component $(S^{[\alpha]}(\vec{x}))_{i_1\cdots i_N}^{j_1\cdots j_N}$ is defined as the partition function of this configuration. An illustration of this component is given below.

\begin{equation}
\hfill
\begin{tikzpicture}[baseline=(current bounding box.center)]
  \draw (-0.5,0.5) to (0.5,0.5);
  \draw (-0.5,1) to (0.5,1);
  \draw (-0.5,2) to (0.5,2);
  \draw (-0.5,2.5) to (0.5,2.5);
  \draw (0,0) to (0,3);
  \node at (-0.5,0.5) [anchor=east] {$i_1$};
 \node at (-0.5,1) [anchor=east] {$i_2$};
  \node at (-0.5,1.5) [anchor=east] {$\cdots$};
 \node at (-0.5,2) [anchor=east] {$i_{N-1}$};
  \node at (-0.5,2.5) [anchor=east] {$i_N$};
   \node at (0.5,0.5) [anchor=west] {$j_1$};
 \node at (0.5,1) [anchor=west] {$j_2$};
  \node at (0.5,1.5) [anchor=west] {$\cdots$};
 \node at (0.5,2) [anchor=west] {$j_{N-1}$};
  \node at (0.5,2.5) [anchor=west] {$j_N$};
\node at (0,0) [anchor=north] {$0$};
\node at (0,3) [anchor=south] {$\alpha$};
\end{tikzpicture}\quad\quad
\end{equation}

More generally, for any $\sigma\in S_N$, $\alpha\in\{0,1\}$, and $\vec{x}=(x_1,\cdots,x_N)\in\mathbb{C}^N$, we define the column operator $S^{[\alpha]}_{\sigma}(\vec{x})\in End(W_1\otimes \cdots \otimes W_N)$ by specifying its components $(S^{[\alpha]}_{\sigma}(\vec{x}))_{i_1\cdots i_N}^{j_1\cdots j_N}$ for any $(i_1,\cdots,i_N),(j_1,\cdots,j_N)\in \{0,1\}^N$. 

To define $(S^{[\alpha]}_{\sigma}(\vec{x}))_{i_1\cdots i_N}^{j_1\cdots j_N}$, consider a column of ordinary vertices whose spectral parameters are given by $x_1,\cdots,x_N$ from bottom to top. We also specify the boundary condition as follows: the top edge is labeled $\alpha$, the bottom edge is labeled $0$, the left edges are labeled $i_{\sigma(1)},\cdots, i_{\sigma(N)}$ (from bottom to top), and the right edges are labeled $j_{\sigma(1)},\cdots,j_{\sigma(N)}$ (from bottom to top). The component $(S^{[\alpha]}_{\sigma}(\vec{x}))_{i_1\cdots i_N}^{j_1\cdots j_N}$ is defined as the partition function of this configuration. A illustration of this component is given below.

\begin{equation}
\hfill
\begin{tikzpicture}[baseline=(current bounding box.center)]
  \draw (-0.5,0.5) to (0.5,0.5);
  \draw (-0.5,1) to (0.5,1);
  \draw (-0.5,2) to (0.5,2);
  \draw (-0.5,2.5) to (0.5,2.5);
  \draw (0,0) to (0,3);
  \node at (-0.5,0.5) [anchor=east] {$i_{\sigma(1)}$};
 \node at (-0.5,1) [anchor=east] {$i_{\sigma(2)}$};
  \node at (-0.5,1.5) [anchor=east] {$\cdots$};
 \node at (-0.5,2) [anchor=east] {$i_{\sigma(N-1)}$};
  \node at (-0.5,2.5) [anchor=east] {$i_{\sigma(N)}$};
   \node at (0.5,0.5) [anchor=west] {$j_{\sigma(1)}$};
 \node at (0.5,1) [anchor=west] {$j_{\sigma(2)}$};
  \node at (0.5,1.5) [anchor=west] {$\cdots$};
 \node at (0.5,2) [anchor=west] {$j_{\sigma(N-1)}$};
  \node at (0.5,2.5) [anchor=west] {$j_{\sigma(N)}$};
\node at (0,0) [anchor=north] {$0$};
\node at (0,3) [anchor=south] {$\alpha$};
\end{tikzpicture}\quad\quad
\end{equation}


\subsection{Permutation graph}\label{Sect.3.3}
In this subsection, we introduce the concept of ``permutation graph'' for free fermionic Boltzmann weights. The permutation graph is a generalization of the $R$-matrix. 

For any two permutations $\rho_1,\rho_2\in S_N$ and any vector of spectral parameters $\vec{x}=(x_1,\cdots,x_N)\in\mathbb{C}^N$, the ``permutation graph'' $R_{\rho_1}^{\rho_2}(\vec{x})$ is an element of $End(W_1\otimes \cdots \otimes W_N)$ as defined below.

We first consider the case where $\rho_1=s_i=(i,i+1)$ and $\rho_2=id$ for some $1\leq i\leq N-1$. We let
\begin{equation}
    R_{s_i}^{id}(\vec{x})=R_{i,  i+1}(x_i,x_{i+1}).
\end{equation}
More generally, for any $\rho_1\in S_N$, we let
\begin{equation}
    R_{\rho_1}^{\rho_1}(\vec{x})=1,\quad R_{\rho_1\circ s_i}^{\rho_1}(\vec{x})=R_{\rho_1(i),\rho_1(i+1)}(x_{\rho_1(i)},x_{\rho_1(i+1)}),
\end{equation}
and recursively for any $\rho_1,\rho_2\in S_N$,
\begin{equation}
    R_{\rho_1 \circ  s_i}^{\rho_2}(\vec{x})=R_{\rho_1}^{\rho_2}(\vec{x})  R_{\rho_1\circ s_i}^{\rho_1}(\vec{x}).
\end{equation}
For general $\rho_1,\rho_2$, $R_{\rho_1}^{\rho_2}(\vec{x})$ can be constructed from the above definition. By the ``$RRR$'' Yang-Baxter equations and the unitarity relation (Theorems \ref{YBE2}-\ref{Unit}), $R_{\rho_1}^{\rho_2}(\vec{x})$ is well-defined.

As a simple example, consider the case where $N=4$, $\rho_1=id$, and $\rho_2=(132)$. Then for any $(i_1,\cdots,i_N),(j_1,\cdots,j_N)\in \{0,1\}^N$, the component $(R_{\rho_1}^{\rho_2}(\vec{x}))_{i_1\cdots i_N}^{j_1 \cdots j_N}$ is given by the partition function of the following configuration. $S$ as shown below is $\Gamma\Gamma$ ice with spectral parameters $x_3,x_1$, and $T$ is $\Gamma\Gamma$ ice with spectral parameters $x_3,x_2$.

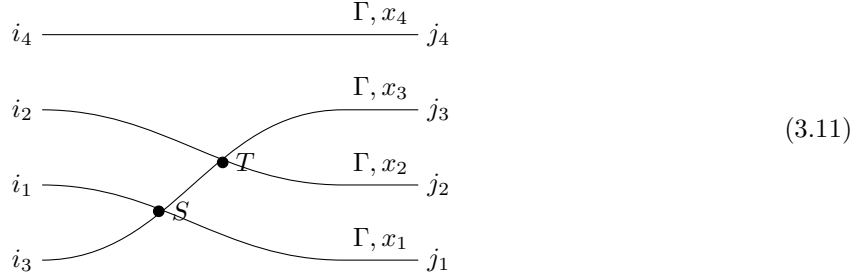
\begin{equation}
\hfill
\begin{tikzpicture}[baseline=(current bounding box.center)]
  \draw (0,3) to (5,3);
  \draw (0,0) to [out=0,in=180] (4,2) to (5,2);
  \draw (0,1) to [out=0,in=180] (4,0) to (5,0);
  \draw (0,2) to [out=0,in=180] (4,1) to (5,1);
  \node at (0,0) [anchor=east] {$i_3$};
  \node at (0,1) [anchor=east] {$i_1$};
  \node at (0,2) [anchor=east] {$i_2$};
  \node at (0,3) [anchor=east] {$i_4$};
  \node at (5,0) [anchor=west] {$j_1$};
  \node at (5,1) [anchor=west] {$j_2$};
  \node at (5,2) [anchor=west] {$j_3$};
  \node at (5,3) [anchor=west] {$j_4$};
  \node at (4.5,3) [anchor=south] {$\Gamma,x_4$};
  \node at (4.5,2) [anchor=south] {$\Gamma, x_3$};
  \node at (4.5,1) [anchor=south] {$\Gamma,x_2$};
  \node at (4.5,0) [anchor=south] {$\Gamma,x_1$};
 \filldraw[black] (1.55,0.65) circle (2pt);
\node at (1.6,0.65) [anchor=west] {$S$};
\filldraw[black] (2.4,1.3) circle (2pt);
\node at (2.45,1.3) [anchor=west] {$T$};
\end{tikzpicture}
\end{equation}

\subsection{$F$-matrix}\label{Sect.3.4}

Based on the permutation graph, we construct the ``$F$-matrix'' as follows. The $F$-matrix introduced in this paper is inspired by the $F$-matrix (also called the ``factorization matrix'') introduced in \cite{MS} for $U_q(\widehat{\mathfrak{sl}}_2)$ $R$-matrix, but works for free fermionic Boltzmann weights.

First we set up some relevant notations. For any permutation $\rho\in S_N$, we define
\begin{eqnarray}\label{I}
    I(\rho)&:=&\{(k_1,\cdots,k_N)\in \{0,1\}^N: 0\leq k_N\leq \cdots\leq k_1\leq 1; \text{ for any }1\leq t \leq N-1, \text{ if }\rho(t)>\rho(t+1),\nonumber \\
    &&\quad \text{ then }k_t>k_{t+1}\},
\end{eqnarray}
\begin{eqnarray}\label{I'}
    I'(\rho)&:=&\{(k_1,\cdots,k_N)\in\{0,1\}^N: 0\leq k_1\leq \cdots\leq k_N\leq 1; \text{ for any }1\leq t \leq N-1, \text{ if }\rho(t)<\rho(t+1),\nonumber\\
    &&\quad \text{ then }k_t<k_{t+1}\}.
\end{eqnarray}
As a simple example, consider the case where $N=2$. When $\rho=id$, we have $I(\rho)=\{(0,0),(1,0),(1,1)\}$ and $I'(\rho)=\{(0,1)\}$. When $\rho=(12)$, we have $I(\rho)=\{(1,0)\}$ and $I'(\rho)=\{(0,0),(0,1),(1,1)\}$.

Now we define the $F$-matrices $F(\vec{x})=F_{1\cdots N}(\vec{x}),F^{*}(\vec{x})=F^{*}_{1\cdots N}(\vec{x})\in End(W_1 \otimes \cdots \otimes W_N)$ by
\begin{equation}
    F(\vec{x}):=\sum_{\rho\in S_N}\sum_{(k_1,\cdots,k_N)\in I(\rho)}\prod_{i=1}^N E_{\rho(i)}^{(k_i,k_i)} R_{id}^{\rho}(\vec{x}),
\end{equation}
\begin{equation}
    F^{*}(\vec{x}):=\sum_{\rho\in S_N}\sum_{(k_1,\cdots,k_N)\in I'(\rho)}R_{\rho}^{id}(\vec{x})\prod_{i=1}^N E_{\rho(i)}^{(k_i,k_i)}.
\end{equation}

Again, consider the case where $N=2$. We have
\begin{equation*}
    F(x_1,x_2)=\sum_{(k_1,k_2)\in I(id)}E_1^{(k_1,k_1)}E_2^{(k_2,k_2)}+E_2^{(1,1)}E_1^{(0,0)}R_{id}^{(12)}(x_1,x_2).
\end{equation*}
We order the basis vectors of $W_1\otimes W_2$ as $|0,0\rangle,|0,1\rangle,|1,0\rangle,|1,1\rangle$. Note that
\begin{equation*}
    R_{id}^{(12)}(x_1,x_2)=\begin{bmatrix}
    a_1(x_2,x_1) & & & \\
    & b_2(x_2,x_1) & c_1(x_2,x_1) &\\
    & c_2(x_2,x_1) &  b_1(x_2,x_1) &\\
    &&& a_2(x_2,x_1)
    \end{bmatrix}.
\end{equation*}
Hence we have
\begin{equation}\label{e1}
   F(x_1,x_2)= \begin{bmatrix}
    1 & & & \\
    & b_2(x_2,x_1) & c_1(x_2,x_1) &\\
    & 0 &  1 &\\
    &&&1
    \end{bmatrix}.
\end{equation}
Similarly, 
\begin{equation*}
    F^{*}(x_1,x_2)=E_1^{(0,0)}E_2^{(1,1)}+\sum_{(k_1,k_2)\in I'((12))}R_{(12)}^{id}(x_1,x_2)E_2^{(k_1,k_1)}E_1^{(k_2,k_2)},
\end{equation*}
which leads to
\begin{equation}\label{e2}
    F^{*}(x_1,x_2)=\begin{bmatrix}
    a_1(x_1,x_2) & &&\\
    & 1 & c_2(x_1,x_2) \\
    & 0 & b_2(x_1,x_2) \\
    &&&a_2(x_1,x_2)
    \end{bmatrix}.
\end{equation}
Using the relations
\begin{equation*}
    b_2(x_2,x_1)c_2(x_1,x_2)+c_1(x_2,x_1)b_2(x_1,x_2)=0,\quad  a_1(x_1,x_2)=1,
\end{equation*}
we have 
\begin{equation}\label{Eq}
    F(x_1,x_2) F^{*}(x_1,x_2)=\begin{bmatrix}
    1 &&&\\
    & b_2(x_2,x_1) &&\\
    && b_2(x_1,x_2)&\\
    &&&a_2(x_1,x_2)
    \end{bmatrix}.
\end{equation}
More generally, for any $\sigma\in S_N$, we can define $F_{\sigma(1)\cdots\sigma(N)}(x_1,\cdots,x_N)$ as follows. We define the permutation operator $P_{1\cdots N}^{\sigma}$ by
\begin{equation*}
    P_{1\cdots N}^{\sigma}|i_1,\cdots,i_N\rangle=|i_{\sigma^{-1}(1)},\cdots, i_{\sigma^{-1}(N)}\rangle.
\end{equation*}
Now we define
\begin{equation}
    F_{\sigma(1)\cdots\sigma(N)}(x_1,\cdots,x_N)=P_{1\cdots N}^{\sigma} F_{1\cdots N}(x_1,\cdots,x_N) P_{1\cdots N}^{\sigma^{-1}}.
\end{equation}
When $N=2$, we have
\begin{equation}\label{e3}
    F_{21}(x_2,x_1)=\begin{bmatrix}
    1 &&&\\
    &1&0&\\
    &c_1(x_1,x_2)&b_2(x_1,x_2)&\\
    &&&1
    \end{bmatrix}.
\end{equation}

Hereafter, we may omit the argument $\vec{x}$ from the $F$-matrix when it is clear from the context.

\subsection{Basic properties of the $F$-matrix}\label{Sect.3.5}

In this subsection, we derive some basic properties of the $F$-matrix introduced in Section \ref{Sect.3.4}. 

The following proposition gives the inverse $F^{-1}$ of $F$ in terms of $F^{*}$.

\begin{proposition}\label{P1}
$\Delta:=F F^{*}$ is a diagonal matrix. The diagonal entries of $\Delta$ are given by
\begin{eqnarray*}
 (\Delta)_{i_1\cdots i_N}^{i_1\cdots i_N}=\prod_{(a,b):i_a=1,i_b=0} b_2(x_a,x_b)\prod_{(a,b):a<b,i_a=1,i_b=1}a_2(x_a,x_b)
\end{eqnarray*}
for every $(i_1,\cdots,i_N)\in\{0,1\}^N$.
\end{proposition}
\begin{remark}
The special case where $N=2$ is computed in (\ref{Eq}).
\end{remark}
\begin{remark}
The proposition implies that $F^{-1}=F^{*}\Delta^{-1}$ if $\Delta$ is invertible.
\end{remark}
\begin{proof}
Let $\sigma_0\in S_N$ be defined by $\sigma_0(i)=N+1-i$ for every $1\leq i\leq N$.
By elementary computations (similar to those in \cite[Proposition 3.2]{ABFR}), we obtain that
\begin{equation}\label{E1}
 \Delta=F F^{*}=\sum_{\rho\in S_N}\sum_{(k_1,\cdots,k_N)\in I(\rho)}\prod_{i=1}^N E_{\rho(i)}^{(k_i,k_i)} R_{\rho\sigma_0}^{\rho}(\vec{x})\prod_{i=1}^N E_{\rho(i)}^{(k_i,k_i)}.
\end{equation}
This shows that $\Delta$ is a diagonal matrix.

Now we compute the component $(\Delta)_{i_1\cdots i_N}^{i_1\cdots i_N}$ for every $(i_1,\cdots,i_N)\in\{0,1\}^N$. For any $\rho\in S_N$, the only element $(k_1,\cdots,k_N)\in I(\rho)$ in the sum of the right hand side of (\ref{E1}) that contributes to $(\Delta)_{i_1\cdots i_N}^{i_1\cdots i_N}$ is determined by $k_t=i_{\rho(t)}$ for every $1\leq t\leq N$. From the definition of $I(\rho)$ in (\ref{I}), in order for this term to be non-vanishing, we mush have 
\begin{eqnarray}\label{E2}
&& 0\leq i_{\rho(N)}\leq \cdots \leq i_{\rho(1)}\leq 1,   \nonumber   \\ 
&&    i_{\rho(t)}=i_{\rho(t+1)}\text{ implies }\rho(t)<\rho(t+1), \text{ for every }1\leq t\leq N-1.
\end{eqnarray}
There is a unique permutation, denoted by $\rho(i_1,\cdots,i_N)\in S_N$, that satisfies the condition (\ref{E2}). 

Therefore we conclude that
\begin{equation}\label{E3}
    (\Delta)_{i_1\cdots i_N}^{i_1\cdots i_N}=(R_{\rho(i_1,\cdots,i_N)\sigma_0}^{\rho(i_1\cdots i_N)}(\vec{x}))_{i_1\cdots i_N}^{i_1\cdots i_N}.
\end{equation}
To compute the right hand side of (\ref{E3}), we note that there is only one admissible state for the corresponding permutation graph, as indicated below (in which $\rho=\rho(i_1,\cdots,i_N)$).

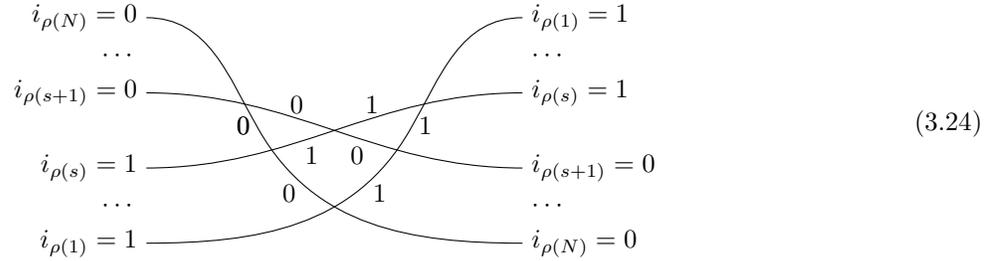
\begin{equation}
\hfill
\begin{tikzpicture}[baseline=(current bounding box.center)]
  \draw (0,3) to [out=0, in=120] (1.5,1.5) to [out=-60, in=180] (5,0);
  \draw (0,0) to [out=0,in=-120] (3.5,1.5) to [out=60, in=180] (5,3);
  \draw (0,1) to [out=0,in=180] (5,2);
  \draw (0,2) to [out=0,in=180] (5,1);
  \node at (0,0.5) [anchor=east] {$\cdots$};
  \node at (0,0) [anchor=east] {$i_{\rho(1)}=1$};
  \node at (0,1) [anchor=east] {$i_{\rho(s)}=1$};
  \node at (0,2) [anchor=east] {$i_{\rho(s+1)}=0$};
    \node at (0,2.5) [anchor=east] {$\cdots$};
  \node at (0,3) [anchor=east] {$i_{\rho(N)}=0$};
  \node at (5,0) [anchor=west] {$i_{\rho(N)}=0$};
  \node at (5,0.5) [anchor=west] {$\cdots$};
  \node at (5,1) [anchor=west] {$i_{\rho(s+1)}=0$};
  \node at (5,2.5) [anchor=west] {$\cdots$};
  \node at (5,2) [anchor=west] {$i_{\rho(s)}=1$};
  \node at (5,3) [anchor=west] {$i_{\rho(1)}=1$};
  \node at (2,1.6) [anchor=south] {$0$};
  \node at (3,1.6) [anchor=south] {$1$};
  \node at (1.5,1.8) [anchor=north east] {$0$};
  \node at (3.5,1.8) [anchor=north west] {$1$};
  \node at (1.5,1.8) [anchor=north east] {$0$};
 \node at (2.2,1.4) [anchor=north] {$1$};
 \node at (2.8,1.4) [anchor=north ] {$0$};
  \node at (1.9,0.9) [anchor=north] {$0$};
 \node at (3.1,0.9) [anchor=north ] {$1$};
\end{tikzpicture}\quad\quad
\end{equation}
The Boltzmann weight of the unique admissible state is (note that there are only crossings of $a_1,a_2$, or $b_2$ patterns as can be seen from the above figure)
\begin{equation}
    \prod_{(a,b): i_a=1,i_b=0}b_2(x_a,x_b)\prod_{(a,b):a<b,i_a=1,i_b=1}a_2(x_a,x_b).
\end{equation}
Hence for every $(i_1,\cdots,i_N)\in\{0,1\}^N$, we have
\begin{eqnarray*}
 (\Delta)_{i_1\cdots i_N}^{i_1\cdots i_N}=\prod_{(a,b):i_a=1,i_b=0} b_2(x_a,x_b)\prod_{(a,b):a<b,i_a=1,i_b=1}a_2(x_a,x_b).
\end{eqnarray*}

\end{proof}

Using similar arguments, we obtain the following proposition on the components of $F$ and $F^{*}$.

\begin{proposition}\label{P2}
For any given $(i_1,\cdots,i_N)\in\{0,1\}^N$, let $\rho=\rho(i_1,\cdots,i_N)\in S_N$ be the unique permutation determined by the following condition
\begin{eqnarray*}
  && 0\leq i_{\rho(N)}\leq \cdots \leq i_{\rho(1)} \leq 1,\\
  && i_{\rho(t)}=i_{\rho(t+1)}\text{ implies }\rho(t)<\rho(t+1), \text{ for every }  1\leq t\leq N-1.
\end{eqnarray*}
Similarly, for any given $(j_1,\cdots,j_N)\in \{0,1\}^N$, let $\rho^{*}=\rho^{*}(j_1,\cdots,j_N)\in S_N$ be the unique permutation determined by the following condition
\begin{eqnarray*}
&& 0\leq j_{\rho^{*}(1)}\leq \cdots \leq j_{\rho^{*}(N)} \leq 1,\\
&& j_{\rho^{*}(t)}=j_{\rho^{*}(t+1)}\text{ implies }\rho^{*}(t)>\rho^{*}(t+1), \text{ for every }1\leq t\leq N-1.
\end{eqnarray*}
Then for any $(i_1,\cdots,i_N),(j_1,\cdots,j_N)\in\{0,1\}^N$, the components of $F$, $F^{*}$ are given by
\begin{equation}
    (F)_{i_1\cdots i_N}^{j_1\cdots j_N}=(R_{id}^{\rho})_{i_1\cdots i_N}^{j_1 \cdots j_N},
\end{equation}
\begin{equation}
    (F^{*})_{i_1 \cdots i_N}^{j_1 \cdots j_N}=(R_{\rho^{*}}^{id})_{i_1\cdots i_N}^{j_1 \cdots j_N}.
\end{equation}
\end{proposition}
\begin{remark}
For the special case where $N=2$, this can be directly checked from (\ref{e1}) and (\ref{e2}).
\end{remark}

\section{Ice model related to Tokuyama's formula}\label{Sect.4}

In this section, based on the concepts introduced in Section \ref{Sect.3}, we apply the method outlined in Section \ref{Sect.1.2} to give a new derivation of the partition function of the lattice model in \cite{BBF}. As discussed in the Introduction, computing the partition function of this lattice model leads to an alternative proof of Tokuyama's formula \cite{Tok}. In Section \ref{Sect.4.1}, we review the setups of the model in \cite{BBF}. Then we present the computation of its partition function in Section \ref{Sect.4.2}.

\subsection{The lattice model}\label{Sect.4.1}

We review the lattice model in \cite{BBF} as follows. Let $\lambda=(\lambda_1,\cdots,\lambda_N)$ be a given partition (meaning that $\lambda_1\geq \cdots\geq\lambda_N\geq 0$), and $\rho=(N-1,\cdots,0)$. Consider a rectangular lattice with $N$ rows and $\lambda_1+N$ columns. The columns are labeled $0,1,\cdots,\lambda_1+N-1$ from right to left, and the rows are labeled $1,2,\cdots,N$ from bottom to top. Note that our notations differ from that of \cite{BBF} in that the ordering of the rows is reversed. The vertices are all $\Gamma$ ice, and the spectral parameter of the vertices in the $i$th row is given by $z_i$ for every $1\leq i\leq N$.

The boundary conditions are specified as follows. On the left and bottom boundaries we assign $0$ ($+$ spin). On the right boundary we assign $1$ ($-$ spin). On the top boundary, we assign $1$ to every column labeled $\lambda_i+N-i$ for $1\leq i\leq N$, and $0$ to the rest of the columns.

Let $z=(z_1,\cdots,z_N)$. We denote by $Z(\mathcal{S}_{\lambda,z})$ the partition function of the above lattice model.

As an illustration, when $N=3$ and $\lambda=(3,1,0)$, the model configuration is given below:

\begin{equation}
\hfill
\begin{tikzpicture}[baseline=(current bounding box.center)]
\draw (0,1)--(7,1);
\draw (0,2)--(7,2);
\draw (0,3)--(7,3);
\draw (1,0.5)--(1,3.5);
\draw (2,0.5)--(2,3.5);
\draw (3,0.5)--(3,3.5);
\draw (4,0.5)--(4,3.5);
\draw (5,0.5)--(5,3.5);
\draw (6,0.5)--(6,3.5);
\filldraw[black] (1,1) circle (1pt);
\filldraw[black] (2,1) circle (1pt);
\filldraw[black] (3,1) circle (1pt);
\filldraw[black] (4,1) circle (1pt);
\filldraw[black] (5,1) circle (1pt);
\filldraw[black] (6,1) circle (1pt);
\filldraw[black] (3,2) circle (1pt);
\filldraw[black] (2,2) circle (1pt);
\filldraw[black] (1,2) circle (1pt);
\filldraw[black] (4,2) circle (1pt);
\filldraw[black] (5,2) circle (1pt);
\filldraw[black] (6,2) circle (1pt);
\filldraw[black] (1,3) circle (1pt);
\filldraw[black] (2,3) circle (1pt);
\filldraw[black] (3,3) circle (1pt);
\filldraw[black] (4,3) circle (1pt);
\filldraw[black] (5,3) circle (1pt);
\filldraw[black] (6,3) circle (1pt);
\node at (0,1) [anchor=east] {$0$};
\node at (0,2) [anchor=east] {$0$};
\node at (0,3) [anchor=east] {$0$};
\node at (7,1) [anchor=west] {$1$};
\node at (7,2) [anchor=west] {$1$};
\node at (7,3) [anchor=west] {$1$};
\node at (-0.5,1) [anchor=east] {$1$};
\node at (-0.5,2) [anchor=east] {$2$};
\node at (-0.5,3) [anchor=east] {$3$};
\node at (-1,1) [anchor=east] {row};
\node at (0.5,1) [anchor=south] {$\Gamma$};
\node at (0.5,2) [anchor=south] {$\Gamma$};
\node at (0.5,3) [anchor=south] {$\Gamma$};
\node at (6.5,1) [anchor=south] {$z_1$};
\node at (6.5,2) [anchor=south] {$z_2$};
\node at (6.5,3) [anchor=south] {$z_3$};
\node at (1,3.5) [anchor=south] {$1$};
\node at (2,3.5) [anchor=south] {$0$};
\node at (3,3.5) [anchor=south] {$0$};
\node at (4,3.5) [anchor=south] {$1$};
\node at (5,3.5) [anchor=south] {$0$};
\node at (6,3.5) [anchor=south] {$1$};
\node at (1,4) [anchor=south] {$5$};
\node at (2,4) [anchor=south] {$4$};
\node at (3,4) [anchor=south] {$3$};
\node at (4,4) [anchor=south] {$2$};
\node at (5,4) [anchor=south] {$1$};
\node at (6,4) [anchor=south] {$0$};
\node at (0,4) [anchor=south] {column};
\node at (1,0.5) [anchor=north] {$0$};
\node at (2,0.5) [anchor=north] {$0$};
\node at (3,0.5) [anchor=north] {$0$};
\node at (4,0.5) [anchor=north] {$0$};
\node at (5,0.5) [anchor=north] {$0$};
\node at (6,0.5) [anchor=north] {$0$};
\end{tikzpicture}
\end{equation}

\subsection{Computation of the partition function}\label{Sect.4.2}

In this subsection, we use the method outlined in Section \ref{Sect.1.2} and the concepts introduced in Section \ref{Sect.3} to provide a new derivation of the partition function of the lattice model as reviewed in Section \ref{Sect.4.1}. The partition function was originally derived in \cite{BBF} using a different approach, which is based on combinatorics of Gelfand-Testlin patterns. The main result is the following theorem.

\begin{theorem}\label{Theorem1}
\begin{equation}
Z(\mathcal{S}_{\lambda,z})=\prod_{1\leq i< j\leq N}(z_j-vz_i)s_{\lambda}(z_1,\cdots,z_N),
\end{equation}
where $s_{\lambda}(z_1,\cdots,z_N)$ is the Schur polynomial.
\end{theorem}

\paragraph{We give a brief outline of the proof as follows.} Instead of writing the partition function in terms of the product of row transfer matrices as usual, we write it in terms of \underline{column operators} as introduced in Section \ref{Sect.3.1}. See (\ref{E4}) below for the concrete expression. Then, using the $F$-matrix as introduced in Section \ref{Sect.3.4} (which is based on the permutation graph introduced in Section \ref{Sect.3.3}), we conjugate each column operator by the $F$-matrix, as in (\ref{Par}) below. Thanks to the conjugation procedure, the components of the \underline{conjugated column operators} have simple explicit forms as given in Proposition \ref{P3}. The computation of these components are based on basic properties of the $F$-matrix as established in Section \ref{Sect.3.5}, as well as the Yang-Baxter equations and properties of permutation graphs. Finally, based on these components, we can write the partition function as a sum over the symmetric group $S_n$, which leads to the proof of Theorem \ref{Theorem1}.

For all computations below, we assume implicitly that all quantities that appear in the denominator are non-zero, since otherwise we can vary $v$ and $z_1,\cdots,z_N$ and obtain the conclusion by continuity. 

Let $\vec{x}=(z_1,\cdots,z_N)$. Below we omit the argument $\vec{x}$ from the notation $S^{[\alpha]}(\vec{x})$.

To compute the partition function $Z(\mathcal{S}_{\lambda,z})$, we first note that it can be written in terms of column operators
\begin{equation}\label{E4}
  Z(\mathcal{S}_{\lambda,z})=\langle 0,\cdots,0|S^{[m_{\lambda_1+N-1}]}\cdots S^{[m_0]}|1,\cdots,1\rangle,
\end{equation}
where for every $0\leq j\leq \lambda_1+N-1$, $m_j=1$ if $j\in\{\lambda_i+N-i:i\in\{1,\cdots,N\}\}$ and $m_j=0$ otherwise. Here, $\langle 0,\cdots,0|$ and $|1,\cdots,1\rangle$ are determined from the left and right boundary conditions. 

Note that by Proposition \ref{P1}, we have
\begin{equation}
    F^{-1}=F^{*}\Delta^{-1}.
\end{equation}
Hence we can write (\ref{E4}) in the following form by conjugating each column operator by the $F$-matrix:
\begin{equation}\label{Par}
    Z(\mathcal{S}_{\lambda,z})=(\langle 0,\cdots,0| F^{-1}) (F S^{[m_{\lambda_1+N-1}]}F^{-1}) \cdots (F S^{[m_0]}F^{-1}) (F|1,\cdots,1\rangle).
\end{equation}
Therefore, it suffices to compute $\langle 0,\cdots,0| F^{-1}$, $F|1,\cdots,1\rangle$ and the components of the conjugated column operators $F S^{[m_{j}]}F^{-1}$ for $0\leq j\leq \lambda_1+N-1$. The motivation for introducing these conjugations is that the components of the conjugated operators can be explicitly computed. 

The following proposition explicitly computes $\langle 0,\cdots,0| F^{-1}$, $F|1,\cdots,1\rangle$, and the components of $F S^{[m]}F^{-1}$ for $m\in\{0,1\}$.

\begin{proposition}\label{P3}
We have
\begin{equation}
    \langle 0,\cdots 0 | F^{-1}=\langle 0,\cdots,0|,
\end{equation}
\begin{equation}
 F|1,\cdots 1 \rangle=|1,\cdots,1\rangle.
\end{equation}
For any $(i_1,\cdots,i_N),(j_1,\cdots,j_N)\in \{0,1\}^N$, we have
\begin{equation}
    (F S^{[0]} F^{-1})_{i_1\cdots i_N}^{j_1 \cdots j_N}=\prod_{t=1}^N \mathbbm{1}_{i_t=j_t}\prod_{t: i_t=1}z_t,
\end{equation}
\begin{eqnarray}
    (F S^{[1]} F^{-1})_{i_1\cdots i_N}^{j_1\cdots j_N}&=&\sum_{m=1}^N \mathbbm{1}_{i_m=0,j_m=1}\prod_{t:1\leq t\leq N, t\neq m}  \mathbbm{1}_{i_t=j_t}\prod_{t: i_t=1,j_t=1}z_t\nonumber\\
    &&\times  \prod_{t:t>m,i_t=1,j_t=1}\frac{z_t-v z_m}{z_m-vz_t}\prod_{t:i_t=0,j_t=0}\frac{z_t-v z_m}{z_m-z_t}.
\end{eqnarray}
Equivalently, we have
\begin{equation}
    F S^{[0]} F^{-1}=\bigotimes_{t\in\{1,\cdots,N\}}\begin{pmatrix}
    1 & 0\\
    0 & z_t
    \end{pmatrix}_t,
\end{equation}
\begin{eqnarray}
    F S^{[1]} F^{-1}=\sum_{m=1}^N \bigotimes_{t\in \{1,\cdots,m-1\}}
    \begin{pmatrix}
    \frac{z_t-vz_m}{z_m-z_t} & 0\\
    0 &z_t
    \end{pmatrix}_t
    \bigotimes \begin{pmatrix}
    0 & 1\\
    0 & 0
    \end{pmatrix}_m
    \bigotimes_{t\in\{m+1,\cdots,N\}}
    \begin{pmatrix}
    \frac{z_t-vz_m}{z_m-z_t} & 0\\
    0 &\frac{z_t(z_t-vz_m)}{z_m-vz_t}
    \end{pmatrix}_t,
\end{eqnarray}
where the basis vectors of each $W_t$ for $t\in\{1,\cdots,N\}$ are ordered as $|0\rangle$, $|1\rangle$.
\end{proposition}

\begin{proof}
We start with the computation of $\langle 0,\cdots,0 | F^{-1}$. For any $(j_1,\cdots,j_N)\in\{0,1\}^N$, as $F^{-1}=F^{*}\Delta^{-1}$ and $\Delta$ is a diagonal matrix, we have
\begin{eqnarray*}
(\langle 0,\cdots ,0| F^{-1} )^{j_1\cdots j_N}&=& \langle 0,\cdots ,0| F^{*}\Delta^{-1}  |j_1,\cdots,j_N\rangle\\
&=& (F^{*})_{0\cdots 0}^{j_1\cdots j_N}(\Delta^{-1})_{j_1\cdots j_N}^{j_1\cdots j_N}.
\end{eqnarray*}
By Proposition \ref{P1},
\begin{eqnarray}\label{Delt}
    (\Delta^{-1})_{j_1\cdots j_N}^{j_1\cdots j_N} &=& \prod_{(a,b):j_a=1,j_b=0} b_2(z_a,z_b)^{-1} \prod_{(a,b):a<b,j_a=1,j_b=1}a_2(z_a,z_b)^{-1}  \nonumber\\
    &=& \prod_{(a,b):j_a=1, j_b=0}\frac{z_b-v z_a}{z_a-z_b} \prod_{(a,b): a<b,j_a=1,j_b=1}\frac{z_b-vz_a}{z_a-vz_b}.
\end{eqnarray}
By Proposition \ref{P2}, 
\begin{equation}
    (F^{*})_{0\cdots 0}^{j_1\cdots j_N}=(R_{\rho^{*}}^{id})_{0\cdots 0}^{j_1\cdots j_N},
\end{equation}
where $\rho^{*}\in S_N$ is the unique permutation determined by
\begin{eqnarray}\label{Con2}
   && 0\leq j_{\rho^{*}(1)}\leq\cdots \leq j_{\rho^{*}(N)}\leq 1,\nonumber\\
   && j_{\rho^{*}(t)}=j_{\rho^{*}(t+1)}\text{ implies }\rho^{*}(t)>\rho^{*}(t+1), \text{ for every }1\leq t\leq N-1.
\end{eqnarray}
By the definition of permutation graph and the Boltzmann weights of the $R$-vertices, in order for $(F^{*})_{0\cdots 0}^{j_1\cdots j_N}$ to be non-vanishing, we necessarily have
\begin{equation}
    j_1=\cdots=j_N=0.
\end{equation}
In this case, $(F^{*})_{0\cdots 0}^{0\cdots 0}=1$, $(\Delta^{-1})_{0\cdots 0}^{0\cdots 0}=1$. Therefore we conclude that
\begin{equation}
    \langle 0,\cdots,0|F^{-1}=\langle 0,\cdots,0|.
\end{equation}

Now we compute $F|1,\cdots,1\rangle$. For any $(i_1,\cdots,i_N)\in\{0,1\}^N$, we have by Proposition \ref{P2},
\begin{eqnarray*}
(F|1,\cdots,1\rangle))_{i_1\cdots i_N}= (F)_{i_1\cdots i_N}^{1\cdots 1}
=(R_{id}^{\rho})_{i_1\cdots i_N}^{1\cdots 1},
\end{eqnarray*}
where $\rho\in S_N$ is the unique permutation that satisfies the condition 
\begin{eqnarray}\label{Con1}
  && 0\leq i_{\rho(N)}\leq \cdots \leq i_{\rho(1)} \leq 1,\nonumber\\
  && i_{\rho(t)}=i_{\rho(t+1)}\text{ implies }\rho(t)<\rho(t+1), \text{ for every }  1\leq t\leq N-1.
\end{eqnarray}
In order for $(F)_{i_1\cdots i_N}^{1\cdots 1}$ to be non-vanishing, necessarily
\begin{equation}
    i_1=\cdots=i_N=1.
\end{equation}
In this case, $\rho=id$, and $(F)_{1\cdots 1}^{1\cdots 1}=1$. Therefore we conclude that
\begin{equation}
    F|1,\cdots 1 \rangle=|1,\cdots,1\rangle.
\end{equation}

Finally, we compute the components of $F S^{[\alpha]} F^{-1}$ for $\alpha\in \{0,1\}$. For any $(i_1,\cdots,i_N),(j_1,\cdots,j_N)\in\{0,1\}^N$, by Proposition \ref{P2}, we have
\begin{eqnarray*}
(F S^{[\alpha]} F^{-1})_{i_1\cdots i_N}^{j_1\cdots j_N} &=&  (F S^{[\alpha]} F^{*})_{i_1\cdots i_N}^{j_1\cdots j_N} (\Delta^{-1})_{j_1\cdots j_N}^{j_1\cdots j_N}\\
&=& (R_{id}^{\rho} S^{[\alpha]} R_{\rho^{*}}^{id})_{i_1\cdots i_N}^{j_1\cdots j_N} (\Delta^{-1})_{j_1\cdots j_N}^{j_1 \cdots j_N},
\end{eqnarray*}
where $\rho,\rho^{*}$ are determined by the conditions (\ref{Con1}) and (\ref{Con2}), respectively.

We note that $(\Delta^{-1})_{j_1\cdots j_N}^{j_1 \cdots j_N}$ is already computed in (\ref{Delt}). Now we compute the component $(R_{id}^{\rho} S^{[\alpha]} R_{\rho^{*}}^{id})_{i_1\cdots i_N}^{j_1\cdots j_N}$. By the Yang-Baxter equations (the ``$RTT$'' version, see Proposition \ref{YBE1}), we can sequentially push the $R$-braids to the right, which leads to
\begin{eqnarray}\label{P}
   (R_{id}^{\rho} S^{[\alpha]} R_{\rho^{*}}^{id})_{i_1\cdots i_N}^{j_1\cdots j_N}&=& (S^{[\alpha]}_{\rho}((z_{\rho(1)},\cdots,z_{\rho(N)})) R_{id}^{\rho} R_{\rho^{*}}^{id})_{i_{1}\cdots i_{N}}^{j_1\cdots j_N}\nonumber\\
   &=& (S^{[\alpha]}_{\rho}((z_{\rho(1)},\cdots,z_{\rho(N)})) R_{\rho^{*}}^{\rho})_{i_{1}\cdots i_{N}}^{j_1\cdots j_N}.
\end{eqnarray}
An illustration of the permutation graph corresponding to the right hand side of (\ref{P}) is given in the following figure.

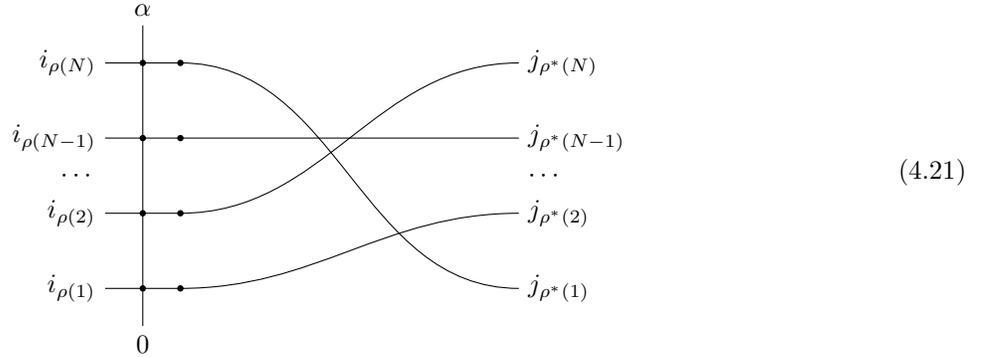
\begin{equation}
\hfill
\begin{tikzpicture}[baseline=(current bounding box.center)]
  \draw (-0.5,0) to (0.5,0);
  \draw (-0.5,1) to (0.5,1);
  \draw (-0.5,2) to (0.5,2);
  \draw (-0.5,3) to (0.5,3);
  \draw (0,-0.5) to (0,3.5);
  \node at (0,3.5) [anchor=south] {$\alpha$};
  \node at (0,-0.5) [anchor=north] {$0$};
  \draw (0.5,3) to [out=0, in=180] (5,0);
  \draw (0.5,0) to [out=0, in=180] (5,1);
  \draw (0.5,1) to [out=0,in=180] (5,3);
  \draw (0.5,2) to [out=0,in=180] (5,2);
  \node at (-0.5,1.5) [anchor=east] {$\cdots$};
  \node at (-0.5,0) [anchor=east] {$i_{\rho(1)}$};
  \node at (-0.5,1) [anchor=east] {$i_{\rho(2)}$};
  \node at (-0.5,2) [anchor=east] {$i_{\rho(N-1)}$};
  \node at (-0.5,3) [anchor=east] {$i_{\rho(N)}$};
  \node at (5,0) [anchor=west] {$j_{\rho^{*}(1)}$};
  \node at (5,1.5) [anchor=west] {$\cdots$};
  \node at (5,1) [anchor=west] {$j_{\rho^{*}(2)}$};
  \node at (5,2) [anchor=west] {$j_{\rho^{*}(N-1)}$};
  \node at (5,3) [anchor=west] {$j_{\rho^{*}(N)}$};
  \filldraw[black] (0.5,0) circle (1pt);
  \filldraw[black] (0.5,1) circle (1pt);
  \filldraw[black] (0.5,2) circle (1pt);
  \filldraw[black] (0.5,3) circle (1pt);
  \filldraw[black] (0,0) circle (1pt);
  \filldraw[black] (0,1) circle (1pt);
  \filldraw[black] (0,2) circle (1pt);
  \filldraw[black] (0,3) circle (1pt);
\end{tikzpicture}\quad\quad
\end{equation}

Now we compute the components $(S_{\rho}^{[\alpha]}((z_{\rho(1)},\cdots,z_{\rho(N)})) R_{\rho^{*}}^{\rho})_{i_{1}\cdots i_{N}}^{j_1\cdots j_N}$. We consider the two cases $\alpha=0$ and $\alpha=1$ separately.

\textbf{First consider the case where $\alpha=0$.} By the definitions of $\rho, \rho^{*}$ and spin conservation, there exists $0\leq s \leq N$, such that
\begin{eqnarray*}
  && i_{\rho(1)}=\cdots=i_{\rho(s)}=1,\quad i_{\rho(s+1)}=\cdots=i_{\rho(N)}=0,\\
  && j_{\rho^{*}(1)}=\cdots=j_{\rho^{*}(N-s)}=0,\quad  j_{\rho^{*}(N-s+1)}=\cdots=j_{\rho^{*}(N)}=1.
\end{eqnarray*}
Moreover, by the conditions (\ref{Con1}) and (\ref{Con2}), we can deduce that
\begin{eqnarray*}
    && \rho(1)<\cdots<\rho(s),\quad \rho(s+1)<\cdots<\rho(N),\\
    && \rho^{*}(1)>\cdots>\rho^{*}(N-s),\quad \rho^{*}(N-s+1)>\cdots>\rho^{*}(N).
\end{eqnarray*}
By spin conservation, we can further deduce that in order for $(S_{\rho}^{[0]}((z_{\rho(1)},\cdots,z_{\rho(N)})) R_{\rho^{*}}^{\rho})_{i_{1}\cdots i_{N}}^{j_1\cdots j_N}$ to be non-vanishing, we necessarily have $j_a=i_a$ for any $1\leq a\leq N$ (which we assume in the following).

It can be checked that there is a unique admissible state of the lattice model corresponding to $(S_{\rho}^{[0]}((z_{\rho(1)},\cdots,z_{\rho(N)})) R_{\rho^{*}}^{\rho})_{i_{1}\cdots i_{N}}^{j_1\cdots j_N}$, as indicated below. 
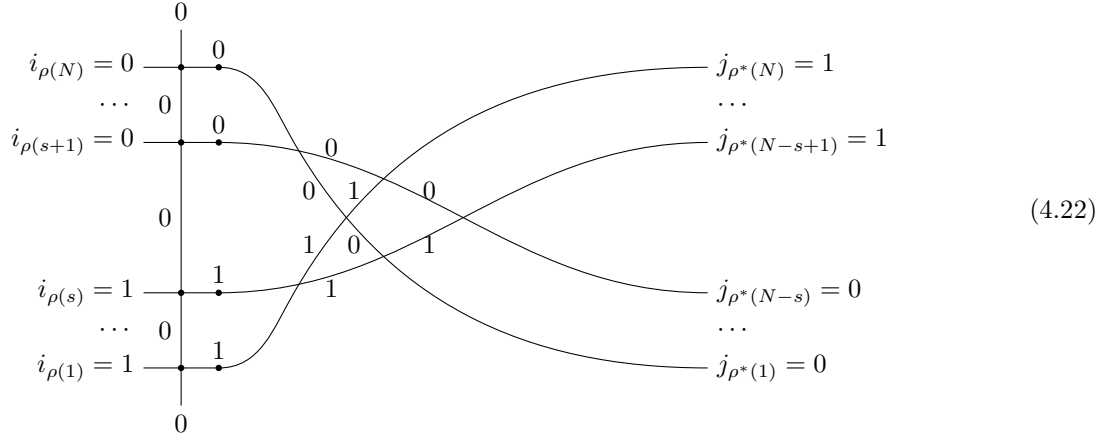
\begin{equation}
\hfill
\begin{tikzpicture}[baseline=(current bounding box.center)]
  \draw (-0.5,0) to (0.5,0);
  \draw (-0.5,1) to (0.5,1);
  \draw (-0.5,3) to (0.5,3);
  \draw (-0.5,4) to (0.5,4);
  \draw (0,-0.5) to (0,4.5);
  \node at (0,4.5) [anchor=south] {$0$};
  \node at (0,-0.5) [anchor=north] {$0$};
  \draw (0.5,0) to [out=0, in=-120] (1.5,1) to [out=60, in=180] (7,4);
  \draw (0.5,1) to [out=0,in=180] (7,3);
   \draw (0.5,4) to [out=0, in=120] (1.5,3) to [out=-60, in=180] (7,0);
  \draw (0.5,3) to [out=0,in=180] (7,1);
  \node at (-0.5,0.5) [anchor=east] {$\cdots$};
  \node at (-0.5,0) [anchor=east] {$i_{\rho(1)}=1$};
  \node at (-0.5,1) [anchor=east] {$i_{\rho(s)}=1$};
  \node at (-0.5,3.5) [anchor=east] {$\cdots$};
  \node at (-0.5,3) [anchor=east] {$i_{\rho(s+1)}=0$};
  \node at (-0.5,4) [anchor=east] {$i_{\rho(N)}=0$};
  \node at (7,0) [anchor=west] {$j_{\rho^{*}(1)}=0$};
  \node at (7,0.5) [anchor=west] {$\cdots$};
  \node at (7,1) [anchor=west] {$j_{\rho^{*}(N-s)}=0$};
  \node at (7,3) [anchor=west] {$j_{\rho^{*}(N-s+1)}=1$};
  \node at (7,4) [anchor=west] {$j_{\rho^{*}(N)}=1$};
   \node at (7,3.5) [anchor=west] {$\cdots$};
  \filldraw[black] (0.5,0) circle (1pt);
  \filldraw[black] (0.5,1) circle (1pt);
  \filldraw[black] (0.5,3) circle (1pt);
  \filldraw[black] (0.5,4) circle (1pt);
  \filldraw[black] (0,0) circle (1pt);
  \filldraw[black] (0,1) circle (1pt);
  \filldraw[black] (0,3) circle (1pt);
  \filldraw[black] (0,4) circle (1pt);
  
  \node at (0.5,0) [anchor=south] {$1$};
  \node at (0.5,1) [anchor=south] {$1$};
  \node at (0.5,4) [anchor=south] {$0$};
  \node at (0.5,3) [anchor=south] {$0$};
  \node at (1.7,1.4) [anchor=south] {$1$};
  \node at (2,1.3) [anchor=north] {$1$};
  \node at (1.7,2.6) [anchor=north] {$0$};
  \node at (2,2.7) [anchor=south] {$0$};
  \node at (2.3,1.4) [anchor=south] {$0$};
  \node at (2.3,2.6) [anchor=north] {$1$};
  \node at (3.3,1.4) [anchor=south] {$1$};
  \node at (3.3,2.6) [anchor=north] {$0$};
   \node at (0,3.5) [anchor=east] {$0$};
  \node at (0,2) [anchor=east] {$0$};
  \node at (0,0.5) [anchor=east] {$0$};
\end{tikzpicture}\quad\quad
\end{equation}
The Boltzmann weight of this unique state is (note that there are only crossings of $a_1$, $a_2$, or $b_2$ pattern in the permutation graph part as in the above figure)
\begin{eqnarray*}
   \prod_{t=1}^N \mathbbm{1}_{i_t=j_t}\prod_{a:i_a=1}b_2(z_a)\prod_{(a,b):a<b,i_a=1,i_b=1}a_2(z_a,z_b)\prod_{(a,b):i_a=1,i_b=0}b_2(z_a,z_b).
\end{eqnarray*}
Hence
\begin{eqnarray*}
    &&(S_{\rho}^{[0]}((z_{\rho(1)},\cdots,z_{\rho(N)})) R_{\rho^{*}}^{\rho})_{i_{1}\cdots i_{N}}^{j_1\cdots j_N}\\
    &=&\prod_{t=1}^N \mathbbm{1}_{i_t=j_t}\prod_{a:i_a=1}b_2(z_a)\prod_{(a,b):a<b,i_a=1,i_b=1}a_2(z_a,z_b)\prod_{(a,b):i_a=1,i_b=0}b_2(z_a,z_b).
\end{eqnarray*}
Therefore, we conclude that 
\begin{equation}
    (F S^{[0]}  F^{-1})_{i_1\cdots i_N}^{j_1 \cdots j_N}=\prod_{t=1}^N \mathbbm{1}_{i_t=j_t}\prod_{a: i_a=1}z_a.
\end{equation}

\textbf{Now we consider the case where $\alpha=1$.} Again, by the definitions of $\rho, \rho^{*}$ and spin conservation, there exists $0\leq s\leq N-1$, such that
\begin{eqnarray*}
  && i_{\rho(1)}=\cdots=i_{\rho(s)}=1,\quad i_{\rho(s+1)}=\cdots=i_{\rho(N)}=0,\\
  && j_{\rho^{*}(1)}=\cdots=j_{\rho^{*}(N-s-1)}=0,\quad  j_{\rho^{*}(N-s)}=\cdots=j_{\rho^{*}(N)}=1.
\end{eqnarray*}
Moreover, by the conditions (\ref{Con1}) and (\ref{Con2}),
\begin{eqnarray*}
    && \rho(1)<\cdots<\rho(s),\quad \rho(s+1)<\cdots<\rho(N),\\
    && \rho^{*}(1)>\cdots>\rho^{*}(N-s-1),\quad \rho^{*}(N-s)>\cdots>\rho^{*}(N).
\end{eqnarray*}

By spin conservation, we can deduce that in order for $(S_{\rho}^{[1]}((z_{\rho(1)},\cdots,z_{\rho(N)})) R_{\rho^{*}}^{\rho})_{i_{1}\cdots i_{N}}^{j_1\cdots j_N}$ to be non-vanishing, there is a unique integer $1\leq m\leq N$, such that $i_m=0$, $j_m=1$, and $j_a=i_a$ for any $a\neq m$ (which we assume in the following).

For any admissible state corresponding to $(S_{\rho}^{[1]}((z_{\rho(1)},\cdots,z_{\rho(N)})) R_{\rho^{*}}^{\rho})_{i_{1}\cdots i_{N}}^{j_1\cdots j_N}$, the spin of some edges has to take a fixed value, as indicated below.
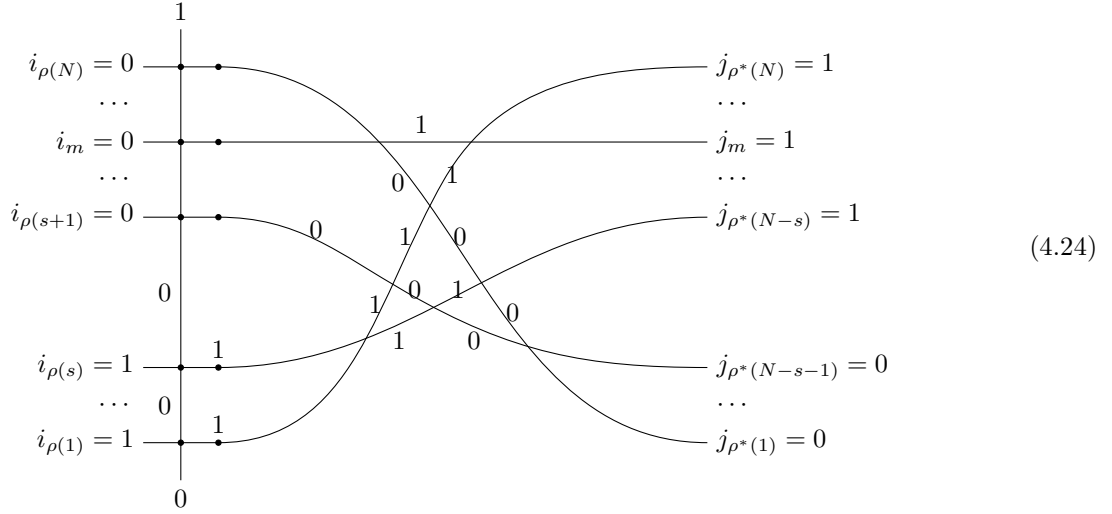
\begin{equation}
\hfill
\begin{tikzpicture}[baseline=(current bounding box.center)]
  \draw (-0.5,0) to (0.5,0);
  \draw (-0.5,1) to (0.5,1);
  \draw (-0.5,3) to (0.5,3);
  \draw (-0.5,4) to (0.5,4);
  \draw (-0.5,5) to (0.5,5);
  \draw (0,-0.5) to (0,5.5);
  \node at (0,5.5) [anchor=south] {$1$};
  \node at (0,-0.5) [anchor=north] {$0$};
  \draw (0.5,0) to [out=0,in=-120] (3.5,3.5) to [out=60,in=180] (7,5);
  \draw (0.5,1) to [out=0,in=180] (7,3);
   \draw (0.5,4) to (7,4);
   \draw (0.5,5) to [out=0,in=180] (7,0);
  \draw (0.5,3) to [out=0,in=150] (3,2) to [out=-30,in=180]  (7,1);
  \node at (-0.5,0.5) [anchor=east] {$\cdots$};
  \node at (-0.5,0) [anchor=east] {$i_{\rho(1)}=1$};
  \node at (-0.5,1) [anchor=east] {$i_{\rho(s)}=1$};
  \node at (-0.5,3.5) [anchor=east] {$\cdots$};
  \node at (-0.5,3) [anchor=east] {$i_{\rho(s+1)}=0$};
  \node at (-0.5,4.5) [anchor=east] {$\cdots$};
  \node at (-0.5,5) [anchor=east] {$i_{\rho(N)}=0$};
    \node at (-0.5,4) [anchor=east] {$i_{m}=0$};
  \node at (7,0) [anchor=west] {$j_{\rho^{*}(1)}=0$};
  \node at (7,0.5) [anchor=west] {$\cdots$};
  \node at (7,1) [anchor=west] {$j_{\rho^{*}(N-s-1)}=0$};
  \node at (7,3) [anchor=west] {$j_{\rho^{*}(N-s)}=1$};
  \node at (7,4) [anchor=west] {$j_m=1$};
   \node at (7,5) [anchor=west] {$j_{\rho^{*}(N)}=1$};
  \node at (7,3.5) [anchor=west] {$\cdots$};
   \node at (7,4.5) [anchor=west] {$\cdots$};
  \filldraw[black] (0.5,0) circle (1pt);
  \filldraw[black] (0.5,1) circle (1pt);
  \filldraw[black] (0.5,3) circle (1pt);
  \filldraw[black] (0.5,4) circle (1pt);
  \filldraw[black] (0,0) circle (1pt);
  \filldraw[black] (0,1) circle (1pt);
  \filldraw[black] (0,3) circle (1pt);
  \filldraw[black] (0,4) circle (1pt);
  \filldraw[black] (0,5) circle (1pt);
  \filldraw[black] (0.5,5) circle (1pt);
  \node at (0,0.5) [anchor=east] {$0$};
  \node at (0,2) [anchor=east] {$0$};
  \node at (0.5,0) [anchor=south] {$1$};
  \node at (0.5,1) [anchor=south] {$1$};
  \node at (2.9,1.6) [anchor=north] {$1$};
  \node at (2.8,1.6) [anchor=south east] {$1$};
  \node at (3.9,1.6) [anchor=north] {$0$};
  \node at (4.2,1.5) [anchor=south west] {$0$};
  \node at (3.2,4) [anchor=south] {$1$};
  \node at (3.4,3.8) [anchor=north west] {$1$};
  \node at (1.8,2.6) [anchor=south] {$0$};
  \node at (3.1,3.7) [anchor=north east] {$0$};
  \node at (3.2,2.5) [anchor=south east] {$1$};
  \node at (3.5,2.5) [anchor=south west] {$0$};
  \node at (2.9,1.8) [anchor=south west] {$0$};
  \node at (3.9,1.8) [anchor=south east] {$1$};
\end{tikzpicture}\quad\quad
\end{equation}

Therefore, we can remove the lines corresponding to $1\leq a\leq N$ such that $i_a=j_a=1$ (together with the corresponding part of the column configuration) up to a factor of
\begin{eqnarray}\label{Q}
    &&\prod_{a:i_a=1,j_a=1}b_2(z_a)\prod_{(a,b):a<b,i_a=1,j_a=1,i_b=1,j_b=1}a_2(z_a,z_b)\nonumber\\
    &\times& \prod_{a: a<m,i_a=1,j_a=1}a_2(z_a,z_m)\prod_{(a,b): i_a=1,j_a=1,i_b=0,j_b=0}b_2(z_a,z_b).
\end{eqnarray}

That is, $(S_{\rho}^{[1]}((z_{\rho(1)},\cdots,z_{\rho(N)})) R_{\rho^{*}}^{\rho})_{i_{1}\cdots i_{N}}^{j_1\cdots j_N}$ is equal to the factor (\ref{Q}) times the partition function of the following configuration
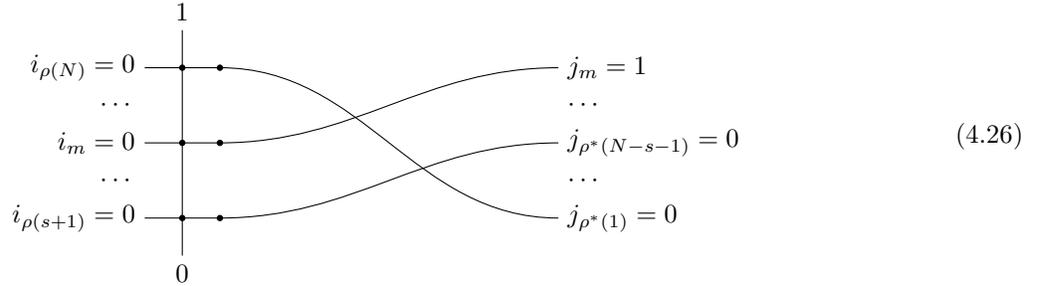
\begin{equation}
\hfill
\begin{tikzpicture}[baseline=(current bounding box.center)]
  \draw (-0.5,0) to (0.5,0);
  \draw (-0.5,1) to (0.5,1);
  \draw (-0.5,2) to (0.5,2);
  \node at (0,2.5) [anchor=south] {$1$};
  \draw (0,-0.5) to (0,2.5);
  \node at (0,-0.5) [anchor=north] {$0$};
  \draw (0.5,0) to [out=0,in=180] (5,1);
  \draw (0.5,1) to [out=0,in=180] (5,2);
  \draw (0.5,2) to [out=0,in=180] (5,0);
  \node at (-0.5,0.5) [anchor=east] {$\cdots$};
  \node at (-0.5,1) [anchor=east] {$i_m=0$};
  \node at (-0.5,0) [anchor=east] {$i_{\rho(s+1)}=0$};
  \node at (-0.5,1.5) [anchor=east] {$\cdots$};
  \node at (-0.5,2) [anchor=east] {$i_{\rho(N)}=0$};
  \node at (5,0) [anchor=west] {$j_{\rho^{*}(1)}=0$};
  \node at (5,0.5) [anchor=west] {$\cdots$};
  \node at (5,1) [anchor=west] {$j_{\rho^{*}(N-s-1)}=0$};
   \node at (5,2) [anchor=west] {$j_m=1$};
     \node at (5,1.5) [anchor=west] {$\cdots$};
  \filldraw[black] (0.5,0) circle (1pt);
  \filldraw[black] (0.5,1) circle (1pt);
  \filldraw[black] (0.5,2) circle (1pt);
  \filldraw[black] (0,0) circle (1pt);
  \filldraw[black] (0,1) circle (1pt);
  \filldraw[black] (0,2) circle (1pt);
\end{tikzpicture}\quad\quad
\end{equation}
By the Yang-Baxter equations (the ``$RTT$'' version, see Proposition \ref{YBE1}), by sequentially pushing the $R$-braids to the left, we derive that the partition function of the above configuration is equal to the partition function of the following configuration
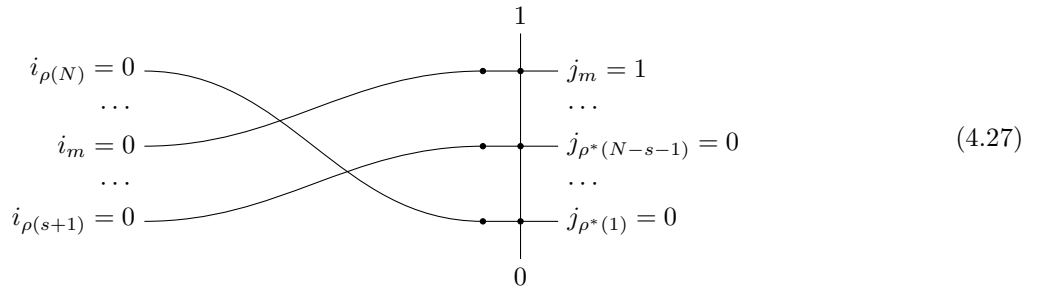
\begin{equation}
\hfill
\begin{tikzpicture}[baseline=(current bounding box.center)]
  \draw (5,0) to (6,0);
  \draw (5,1) to (6,1);
  \draw (5,2) to (6,2);
  \node at (5.5,2.5) [anchor=south] {$1$};
  \draw (5.5,-0.5) to (5.5,2.5);
  \node at (5.5,-0.5) [anchor=north] {$0$};
  \draw (0.5,0) to [out=0,in=180] (5,1);
  \draw (0.5,1) to [out=0,in=180] (5,2);
  \draw (0.5,2) to [out=0,in=180] (5,0);
  \node at (0.5,0.5) [anchor=east] {$\cdots$};
  \node at (0.5,1) [anchor=east] {$i_m=0$};
  \node at (0.5,0) [anchor=east] {$i_{\rho(s+1)}=0$};
  \node at (0.5,1.5) [anchor=east] {$\cdots$};
  \node at (0.5,2) [anchor=east] {$i_{\rho(N)}=0$};
  \node at (6,0) [anchor=west] {$j_{\rho^{*}(1)}=0$};
  \node at (6,0.5) [anchor=west] {$\cdots$};
  \node at (6,1) [anchor=west] {$j_{\rho^{*}(N-s-1)}=0$};
   \node at (6,2) [anchor=west] {$j_m=1$};
     \node at (6,1.5) [anchor=west] {$\cdots$};
  \filldraw[black] (5.5,0) circle (1pt);
  \filldraw[black] (5.5,1) circle (1pt);
  \filldraw[black] (5.5,2) circle (1pt);
  \filldraw[black] (5,0) circle (1pt);
  \filldraw[black] (5,1) circle (1pt);
  \filldraw[black] (5,2) circle (1pt);
\end{tikzpicture}\quad\quad
\end{equation}
Note that there is only one admissible state of this configuration, as indicated below
\begin{equation}
\hfill
\begin{tikzpicture}[baseline=(current bounding box.center)]
  \draw (5,0) to (6,0);
  \draw (5,1) to (6,1);
  \draw (5,2) to (6,2);
  \node at (5.5,2.5) [anchor=south] {$1$};
  \draw (5.5,-0.5) to (5.5,2.5);
  \node at (5.5,-0.5) [anchor=north] {$0$};
  \draw (0.5,0) to [out=0,in=180] (5,1);
  \draw (0.5,1) to [out=0,in=180] (5,2);
  \draw (0.5,2) to [out=0,in=180] (5,0);
  \node at (0.5,0.5) [anchor=east] {$\cdots$};
  \node at (0.5,1) [anchor=east] {$i_m=0$};
  \node at (0.5,0) [anchor=east] {$i_{\rho(s+1)}=0$};
  \node at (0.5,1.5) [anchor=east] {$\cdots$};
  \node at (0.5,2) [anchor=east] {$i_{\rho(N)}=0$};
  \node at (6,0) [anchor=west] {$j_{\rho^{*}(1)}=0$};
  \node at (6,0.5) [anchor=west] {$\cdots$};
  \node at (6,1) [anchor=west] {$j_{\rho^{*}(N-s-1)}=0$};
   \node at (6,2) [anchor=west] {$j_m=1$};
     \node at (6,1.5) [anchor=west] {$\cdots$};
  \filldraw[black] (5.5,0) circle (1pt);
  \filldraw[black] (5.5,1) circle (1pt);
  \filldraw[black] (5.5,2) circle (1pt);
  \filldraw[black] (5,0) circle (1pt);
  \filldraw[black] (5,1) circle (1pt);
  \filldraw[black] (5,2) circle (1pt);
  \node at (4.9,0) [anchor=south] {$0$};
  \node at (2.6,0.8) [anchor=south west] {$0$};
  \node at (4.9,1) [anchor=south] {$0$};
  \node at (4.9,2) [anchor=south] {$0$};
  \node at (5.4,0.5) [anchor=west] {$0$};
  \node at (5.4,1.5) [anchor=west] {$0$};
\end{tikzpicture}\quad\quad
\end{equation}
The Boltzmann weight of this state is
\begin{equation}
    c_2(z_m)=1.
\end{equation}
Hence we have
\begin{eqnarray*}
 &&(S_{\rho}^{[1]}((z_{\rho(1)},\cdots,z_{\rho(N)})) R_{\rho^{*}}^{\rho})_{i_1\cdots i_{N}}^{j_1\cdots j_N}\\
 &=& \prod_{a:i_a=1,j_a=1}b_2(z_a)\prod_{(a,b):a<b,i_a=1,j_a=1,i_b=1,j_b=1}a_2(z_a,z_b)\\
    &\times& \prod_{a: a<m,i_a=1,j_a=1}a_2(z_a,z_m)\prod_{(a,b): i_a=1,j_a=1,i_b=0,j_b=0}b_2(z_a,z_b).
\end{eqnarray*}

Therefore, we conclude that 
\begin{eqnarray}
    (F S^{[1]} F^{-1})_{i_1\cdots i_N}^{j_1\cdots j_N}&=&\sum_{m=1}^N \mathbbm{1}_{i_m=0,j_m=1}\prod_{t:1\leq t\leq N, t\neq m}  \mathbbm{1}_{i_t=j_t}\prod_{a: i_a=1,j_a=1}z_a\nonumber\\
    &&\times  \prod_{a:a>m,i_a=1,j_a=1}\frac{z_a-v z_m}{z_m-vz_a}\prod_{a:i_a=0,j_a=0}\frac{z_a-v z_m}{z_m-z_a}.
\end{eqnarray}

\end{proof}

Now we finish the proof of Theorem \ref{Theorem1} based on Proposition \ref{P3}.

\begin{proof}[Proof of Theorem \ref{Theorem1}]
Note that by Proposition \ref{P3} and (\ref{Par}), $Z(\mathcal{S}_{\lambda,z})$ can be written as the sum of the following terms
\begin{equation}
    (F S^{[m_{\lambda_1+N-1}]} F^{-1})_{0\cdots 0}^{i_1^{(\lambda_1+N-1)}\cdots i_N^{(\lambda_1+N-1)}}\cdots (F S^{[m_1]} F^{-1})_{i_1^{(2)}\cdots i_N^{(2)}}^{i_1^{(1)}\cdots i_N^{(1)}}(F S^{[m_0]} F^{-1})_{i_1^{(1)}\cdots i_N^{(1)}}^{1\cdots 1},
\end{equation}
where $(i_1^{(l)},\cdots,i_N^{(l)})\in\{0,1\}^N$ for $0\leq l\leq \lambda_1+N-1$ satisfy the following condition: if $m_l=0$, then $i_k^{(l)}=i_k^{(l+1)}$ for every $1\leq k\leq N$; if $m_l=1$, then there is a unique index $\alpha_l\in\{1,2,\cdots,N\}$ such that $i_{\alpha_l}^{(l)}=1$, $i_{\alpha_l}^{(l+1)}=0$, and $i_k^{(l)}=i_k^{(l+1)}$ for every $k\neq \alpha_l$. Here, we have assumed that $(i_1^{(\lambda_1+N)},\cdots,i_N^{(\lambda_1+N)})=(0,\cdots,0)$ and $(i_1^{(0)},\cdots,i_N^{(0)})=(1,\cdots,1)$.

Let $\beta_t:=\alpha_{\lambda_t+N-t}$ for every $1\leq t\leq N$. Note that $(\beta_1,\cdots,\beta_N)$ corresponds to a permutation $\sigma\in S_N$ such that $\sigma(t)=\beta_t$ for every $1\leq t\leq N$. Based on this observation and Proposition \ref{P3}, we can deduce that
\begin{eqnarray*}
Z(\mathcal{S}_{\lambda,z}) &=& \sum_{\sigma\in S_N} \prod_{i=1}^N z_{\sigma(i)}^{\lambda_i+N-i}\prod_{(i,j):1\leq i<j\leq N,\sigma(i)>\sigma(j)}\frac{z_{\sigma(i)}-vz_{\sigma(j)}}{z_{\sigma(j)}-vz_{\sigma(i)}}\prod_{(i,j):1\leq i<j\leq N}\frac{z_{\sigma(j)}-v z_{\sigma(i)}}{z_{\sigma(i)}-z_{\sigma(j)}}.
\end{eqnarray*}

Note that for any $\sigma\in S_N$,
\begin{eqnarray*}
&&\prod_{(i,j):1\leq i<j\leq N,\sigma(i)>\sigma(j)}\frac{z_{\sigma(i)}-vz_{\sigma(j)}}{z_{\sigma(j)}-vz_{\sigma(i)}}\prod_{(i,j):1\leq i<j\leq N}(z_{\sigma(j)}-v z_{\sigma(i)})\\
&=& \prod_{(i,j):1\leq i<j\leq N,\sigma(i)<\sigma(j)}(z_{\sigma(j)}-v z_{\sigma(i)})\prod_{(i,j):1\leq i<j\leq N,\sigma(i)>\sigma(j)}(z_{\sigma(i)}-v z_{\sigma(j)})\\
&=& \prod_{1\leq i<j\leq N}(z_j-vz_i).
\end{eqnarray*}

Hence we have
\begin{eqnarray*}
Z(\mathcal{S}_{\lambda,z}) &=& \prod_{1\leq i<j\leq N}(z_j-vz_i)\sum_{\sigma\in S_N}\prod_{i=1}^N z_{\sigma(i)}^{\lambda_i+N-i}\prod_{(i,j):1\leq i<j\leq N}(z_{\sigma(i)}-z_{\sigma(j)})^{-1}\\
&=& \prod_{1\leq i<j\leq N}(z_j-vz_i)\frac{\sum_{\sigma\in S_N}(-1)^{inv(\sigma)}z_{\sigma(i)}^{\lambda_i+N-i}}{\prod_{1\leq i<j\leq N}(z_i-z_j)}\\
&=& \prod_{1\leq i<j\leq N}(z_j-vz_i) s_{\lambda}(z_1,\cdots,z_N),
\end{eqnarray*}
where $inv(\sigma):=\{(i,j):1\leq i<j\leq N,\sigma(i)>\sigma(j)\}$ is the number of inversions of $\sigma\in S_N$.
\end{proof}

\section{Extension to lattice models related to Cartan types B and C}\label{Sect.5}

In this section, we generalize the concepts in Section \ref{Sect.3} to lattice models that are related to Cartan types B and C. Examples of such models include the models in \cite{BBCG,Iva}. In such models, both $\Gamma$ and $\Delta$ vertices appear, and there are U-turn cap vertices on the right boundary. The generalized concepts in this section will be used in Section \ref{Sect.6} to compute the partition function of the lattice model in \cite{BBCG}.

The main difference between this section and Section \ref{Sect.3} is that here we need to distinguish between $\Gamma$ ice and $\Delta$ ice. As in Section \ref{Sect.3}, we fix a positive integer $N$. To each site $i$ (where $1\leq i\leq N$), we associate a number $\epsilon_i\in\{1,-1\}$) to indicate $\Gamma$\slash $\Delta$ type, as detailed below. 

We make the following modifications to the setups in Section \ref{Sect.3.1}. For any $a,b,c,d\in\{0,1\}$, any $x_i,x_j\in\mathbb{C}$, and any $\epsilon_i,\epsilon_j\in\{1,-1\}$, we denote by $R(a,b,c,d;x_i,x_j;\epsilon_i,\epsilon_j)$ the Boltzmann weight of the following $R$-vertex (with spectral parameters $x_i,x_j$ and $\Gamma$\slash $\Delta$ type determined by $\epsilon_i,\epsilon_j$). The determination of the $\Gamma$\slash $\Delta$ type of the $R$-vertex is: if $(\epsilon_i,\epsilon_j)=(1,1)$, it is $\Gamma\Gamma$ ice; if $(\epsilon_i,\epsilon_j)=(-1,-1)$, it is $\Delta\Delta$ ice; if $(\epsilon_i,\epsilon_j)=(1,-1)$, it is $\Gamma\Delta$ ice; if $(\epsilon_i,\epsilon_j)=(-1,1)$, it is $\Delta\Gamma$ ice. 
\begin{equation}
\begin{tikzpicture}[scale=0.7]
\draw (0,0) to [out = 0, in = 180] (2,2);
\draw (0,2) to [out = 0, in = 180] (2,0);
\draw[fill=white] (0,0) circle (.35);
\draw[fill=white] (0,2) circle (.35);
\draw[fill=white] (2,0) circle (.35);
\draw[fill=white] (2,2) circle (.35);
\node at (0,0) {$a$};
\node at (0,2) {$b$};
\node at (2,2) {$c$};
\node at (2,0) {$d$};
\path[fill=white] (1,1) circle (.3);
\node at (1,1) {$R_{x_i,x_j}$};
\end{tikzpicture}
\end{equation}
For any two distinct positive integers $i,j$, we define the $R$-matrix $R_{i,j}(x_i,x_j;\epsilon_i,\epsilon_j)$ that acts on $W_i \otimes W_j$ (with spectral parameters  $x_i,x_j$) by
\begin{equation}
    R_{i,j}(x_i,x_j;\epsilon_i,\epsilon_j)=\sum_{a,b,c,d\in\{0,1\}}R(a,b,c,d;x_i,x_j;\epsilon_i,\epsilon_j) E_i^{(a,c)} E_j^{(b,d)}.
\end{equation}
We also denote $R(x_1,x_2;\epsilon_1,\epsilon_2):=R_{12}(x_1,x_2;\epsilon_1,\epsilon_2)$.

For an ordinary vertex with $\Gamma$\slash $\Delta$ type determined by $\epsilon_i$ ($\epsilon_i=1$ means $\Gamma$ ice and $\epsilon_i=-1$ means $\Delta$ ice) and spectral parameter $x_i$, we denote by $a_1(x_i;\epsilon_i)$ its Boltzmann weight for the $a_1$ state (see Figures \ref{Figure2.1}-\ref{Figure2.2}), and similarly for the other states. For an $R$-vertex with $\Gamma$\slash $\Delta$ type determined by $\epsilon_i,\epsilon_j$ and spectral parameters $x_i,x_j$, we denote by $a_1(x_i,x_j;\epsilon_i,\epsilon_j)$ it Boltzmann weight for the $a_1$ state (see Figures \ref{Figure2.3}-\ref{Figure2.6}), and similarly for the other states.

The parallels of Sections \ref{Sect.3.2}-\ref{Sect.3.4} are given in Sections \ref{Sect.5.1}-\ref{Sect.5.3} below. 

\subsection{Column operator}\label{Sect.5.1}

For any $\alpha\in\{0,1\}$, $\vec{x}=(x_1,\cdots,x_N)\in\mathbb{C}^N$, and $\vec{\epsilon}=(\epsilon_1,\cdots,\epsilon_N)\in\{1,-1\}^N$, the column operator $S^{[\alpha]}(\vec{x};\vec{\epsilon})\in End(W_1\otimes \cdots \otimes W_N)$ is defined by specifying its components $(S^{[\alpha]}(\vec{x};\vec{\epsilon}))_{i_1\cdots i_N}^{j_1\cdots j_N}$ for any $(i_1,\cdots,i_N),(j_1,\cdots,j_N)\in\{0,1\}^N$.

We define $(S^{[\alpha]}(\vec{x};\vec{\epsilon}))_{i_1\cdots i_N}^{j_1\cdots j_N}$ as follows. Consider a column of ordinary vertices whose $\Gamma$\slash $\Delta$ types and spectral parameters are given by $\epsilon_1,\cdots,\epsilon_N$ and $x_1,\cdots,x_N$ from bottom to top. The boundary conditions are specified as follows: the top edge is labeled $\alpha$, the bottom edge is labeled $0$, the left edges are labeled $i_1,\cdots,i_N$ (from bottom to top), and the right edges are labeled $j_1,\cdots,j_N$ (from bottom to top). The component $(S^{[\alpha]}(\vec{x};\vec{\epsilon}))_{i_1\cdots i_N}^{j_1\cdots j_N}$ is defined as the partition function of this configuration.

More generally, for any $\sigma\in S_N$, $\alpha\in\{0,1\}$, $\vec{x}=(x_1,\cdots,x_N)\in\mathbb{C}^N$, and $\vec{\epsilon}=(\epsilon_1,\cdots,\epsilon_N)\in\{1,-1\}^N$, we define the column operator $S^{[\alpha]}_{\sigma}(\vec{x};\vec{\epsilon})\in End(W_1\otimes \cdots \otimes W_N)$ by specifying its components $(S^{[\alpha]}_{\sigma}(\vec{x};\vec{\epsilon}))_{i_1\cdots i_N}^{j_1\cdots j_N}$ for any $(i_1,\cdots,i_N),(j_1,\cdots,j_N)\in \{0,1\}^N$. To define $(S^{[\alpha]}_{\sigma}(\vec{x};\vec{\epsilon}))_{i_1\cdots i_N}^{j_1\cdots j_N}$, consider a column of ordinary vertices whose $\Gamma$\slash $\Delta$ types and spectral parameters are given by $\epsilon_1,\cdots,\epsilon_N$ and $x_1,\cdots,x_N$ from bottom to top. We also specify the boundary condition as follows: the top edge is labeled $\alpha$, the bottom edge is labeled $0$, the left edges are labeled $i_{\sigma(1)},\cdots, i_{\sigma(N)}$ (from bottom to top), and the right edges are labeled $j_{\sigma(1)},\cdots,j_{\sigma(N)}$ (from bottom to top). The component $(S^{[\alpha]}_{\sigma}(\vec{x};\vec{\epsilon}))_{i_1\cdots i_N}^{j_1\cdots j_N}$ is defined as the partition function of this configuration.

\subsection{Permutation graph}\label{Sect.5.2}

For any two permutations $\rho_1,\rho_2\in S_N$, any two vectors $\vec{x}=(x_1,\cdots,x_N)\in\mathbb{C}^N$ and $\vec{\epsilon}=(\epsilon_1,\cdots,\epsilon_N)  \in\{1,-1\}^N$, the ``permutation graph'' $R_{\rho_1}^{\rho_2}(\vec{x};\vec{\epsilon})$ is an element of $End(W_1\otimes \cdots \otimes W_N)$ as defined below.

For the case where $\rho_1=s_i=(i,i+1)$ and $\rho_2=id$ for some $1\leq i\leq N-1$, we let
\begin{equation}
    R_{s_i}^{id}(\vec{x};\vec{\epsilon})=R_{i(i+1)}(x_i,x_{i+1};\epsilon_i,\epsilon_{i+1}).
\end{equation}
For general $\rho_1\in S_N$, we let
\begin{equation}
    R_{\rho_1}^{\rho_1}(\vec{x};\vec{\epsilon})=1,\quad R_{\rho_1\circ s_i}^{\rho_1}(\vec{x};\vec{\epsilon})=R_{\rho_1(i),\rho_1(i+1)}(x_{\rho_1(i)},x_{\rho_1(i+1)};\epsilon_{\rho_1(i)},\epsilon_{\rho_1(i+1)}),
\end{equation}
and recursively for any $\rho_1,\rho_2\in S_N$,
\begin{equation}
    R_{\rho_1 \circ  s_i}^{\rho_2}(\vec{x};\vec{\epsilon})=R_{\rho_1}^{\rho_2}(\vec{x};\vec{\epsilon})  R_{\rho_1\circ s_i}^{\rho_1}(\vec{x};\vec{\epsilon}).
\end{equation}
For general $\rho_1,\rho_2$, $R_{\rho_1}^{\rho_2}(\vec{x};\vec{\epsilon})$ can be constructed from the above definition. By the ``$RRR$'' Yang-Baxter equations and the unitarity relation (Theorems \ref{YBE2}-\ref{Unit}), $R_{\rho_1}^{\rho_2}(\vec{x};\vec{\epsilon})$ is well-defined.

\subsection{$F$-matrix}\label{Sect.5.3}

Now we define the $F$-matrix. For any $\rho\in S_N$, $I(\rho)$ and $I'(\rho)$ are defined as in Section \ref{Sect.3.4}. The $F$-matrices $F(\vec{x};\vec{\epsilon}),F^{*}(\vec{x};\vec{\epsilon})\in End(W_1 \otimes \cdots \otimes W_N)$ are defined by
\begin{equation}
    F(\vec{x};\vec{\epsilon}):=\sum_{\rho\in S_N}\sum_{(k_1,\cdots,k_N)\in I(\rho)}\prod_{i=1}^N E_{\rho(i)}^{(k_i,k_i)} R_{id}^{\rho}(\vec{x};\vec{\epsilon}),
\end{equation}
\begin{equation}
    F^{*}(\vec{x};\vec{\epsilon}):=\sum_{\rho\in S_N}\sum_{(k_1,\cdots,k_N)\in I'(\rho)}R_{\rho}^{id}(\vec{x};\vec{\epsilon})\prod_{i=1}^N E_{\rho(i)}^{(k_i,k_i)}.
\end{equation}

Hereafter, the arguments $\vec{x},\vec{\epsilon}$ may be omitted from the $F$-matrix when they are clear from the context. Proposition \ref{P2} directly generalizes to the setting here (with the newly defined permutation graph and $F$-matrix), and Proposition \ref{P1} generalizes to the following proposition.

\begin{proposition}
$\Delta:=F F^{*}$ is a diagonal matrix. The diagonal entries of $\Delta$ are given by
\begin{eqnarray*}
 (\Delta)_{i_1\cdots i_N}^{i_1\cdots i_N}=\prod_{(a,b):i_a=1,i_b=0} b_2(x_a,x_b;\epsilon_a,\epsilon_b)\prod_{(a,b):a<b,i_a=1,i_b=1}a_2(x_a,x_b;\epsilon_a,\epsilon_b)
\end{eqnarray*}
for every $(i_1,\cdots,i_N)\in\{0,1\}^N$.
\end{proposition}

\section{Ice model representing a Whittaker function on the metaplectic double cover of $\mathrm{Sp}(2r,F)$}\label{Sect.6}

In this section, we apply our method to compute the partition function of the lattice model introduced in \cite{BBCG}. As reviewed in the Introduction, the partition function of this model represents a Whittaker function on the metaplectic double cover of $\mathrm{Sp}(2r,F)$ with $F$ a non-archimedean local field. The computation of this partition function is more involved than that in Section \ref{Sect.4}, as both types of ordinary vertices ($\Gamma$ ice and $\Delta$ ice) are involved in the model, and there are cap vertices connecting adjacent rows of $\Delta$ ice and $\Gamma$ ice on the right boundary.

\subsection{The lattice model}\label{Sect.6.1}

In this subsection, we briefly introduce the lattice model in \cite{BBCG}. Let $\lambda=(\lambda_1,\cdots,\lambda_{r})$ be a given partition (meaning that $\lambda_1\geq\cdots\geq \lambda_N\geq 0$). Consider a rectangular lattice with $2r$ rows and $\lambda_1+r$ columns. The columns are labeled $\frac{1}{2},\frac{3}{2},\cdots,\lambda_1+r-\frac{1}{2}$ from right to left, and the rows are labeled $1,2,\cdots, 2r$ from bottom to top. Note that the ordering of the rows here is reversed from that in \cite{BBCG}. Every odd-numbered row is a row of $\Gamma$ ice, and every even-numbered row is a row of $\Delta$ ice. For every $1\leq i\leq r$, the spectral parameter of the vertices in the $i$th row of $\Gamma$ ice is $z_i^{-1}$, and that of the vertices in the $i$th row of $\Delta$ ice is $z_i$.

The boundary conditions are given as follows. On the left boundary, we assign $0$ ($+$ spin) to each row (note that this is different from the notations in \cite{BBCG} as we switch the signs of spins on horizontal edges of $\Delta$ ice); on the bottom boundary, we assign $0$ to each boundary edge; on the top boundary, we assign $1$ ($-$ spin) to each column labeled $\lambda_i+r+\frac{1}{2}-i$ for every $1\leq i \leq r$, and $0$ to the rest of the columns; on the right, the $i$th row of $\Gamma$ ice and the $i$th row of $\Delta$ ice are connected by a cap vertex with spectral parameter $z_i$. Recall that the Boltzmann weights of the caps are given in Figure \ref{Figure2.7}.

Let $z=(z_1,\cdots,z_r)$. Hereafter, we denote by $Z(\mathcal{T}_{\lambda,z})$ the partition function of the above lattice model. 

As a simple example, when $r=2$, $\lambda=(2,1)$, the model configuration is shown as below.

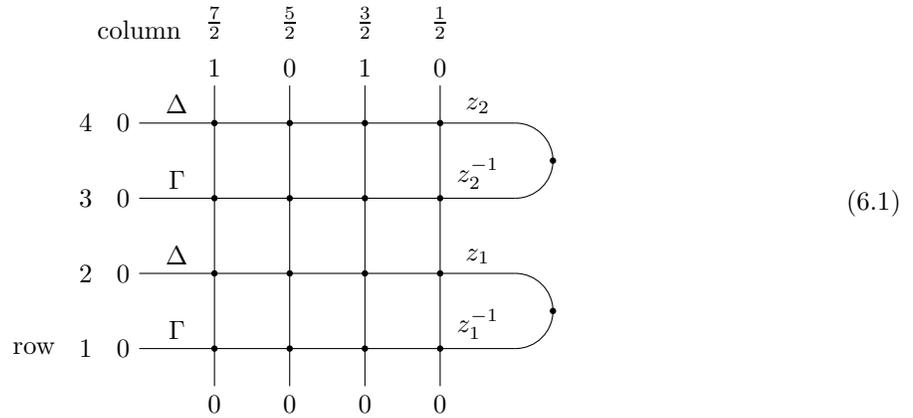
\begin{equation}
\begin{tikzpicture}[baseline=(current bounding box.center)]
\draw (0,1)--(5,1);
\draw (0,2)--(5,2);
\draw (0,3)--(5,3);
\draw (0,4)--(5,4);
\draw (1,0.5)--(1,4.5);
\draw (2,0.5)--(2,4.5);
\draw (3,0.5)--(3,4.5);
\draw (4,0.5)--(4,4.5);
\filldraw[black] (1,1) circle (1pt);
\filldraw[black] (2,1) circle (1pt);
\filldraw[black] (3,1) circle (1pt);
\filldraw[black] (4,1) circle (1pt);
\filldraw[black] (3,2) circle (1pt);
\filldraw[black] (2,2) circle (1pt);
\filldraw[black] (1,2) circle (1pt);
\filldraw[black] (4,2) circle (1pt);
\filldraw[black] (1,3) circle (1pt);
\filldraw[black] (2,3) circle (1pt);
\filldraw[black] (3,3) circle (1pt);
\filldraw[black] (4,3) circle (1pt);
\filldraw[black] (1,4) circle (1pt);
\filldraw[black] (2,4) circle (1pt);
\filldraw[black] (3,4) circle (1pt);
\filldraw[black] (4,4) circle (1pt);
\filldraw[black] (5.5,1.5) circle (1pt);
\filldraw[black] (5.5,3.5) circle (1pt);
\draw (5,1) arc(-90:90:0.5);
\draw (5,3) arc(-90:90:0.5);
\node at (0,1) [anchor=east] {$0$};
\node at (0,2) [anchor=east] {$0$};
\node at (0,3) [anchor=east] {$0$};
\node at (0,4) [anchor=east] {$0$};
\node at (-0.5,1) [anchor=east] {$1$};
\node at (-0.5,2) [anchor=east] {$2$};
\node at (-0.5,3) [anchor=east] {$3$};
\node at (-0.5,4) [anchor=east] {$4$};
\node at (-1,1) [anchor=east] {row};
\node at (0.5,1) [anchor=south] {$\Gamma$};
\node at (0.5,2) [anchor=south] {$\Delta$};
\node at (0.5,3) [anchor=south] {$\Gamma$};
\node at (0.5,4) [anchor=south] {$\Delta$};
\node at (4.5,1) [anchor=south] {$z_1^{-1}$};
\node at (4.5,2) [anchor=south] {$z_1$};
\node at (4.5,3) [anchor=south] {$z_2^{-1}$};
\node at (4.5,4) [anchor=south] {$z_2$};
\node at (1,4.5) [anchor=south] {$1$};
\node at (2,4.5) [anchor=south] {$0$};
\node at (3,4.5) [anchor=south] {$1$};
\node at (4,4.5) [anchor=south] {$0$};
\node at (1,5) [anchor=south] {$\frac{7}{2}$};
\node at (2,5) [anchor=south] {$\frac{5}{2}$};
\node at (3,5) [anchor=south] {$\frac{3}{2}$};
\node at (4,5) [anchor=south] {$\frac{1}{2}$};
\node at (0,5) [anchor=south] {column};
\node at (1,0.5) [anchor=north] {$0$};
\node at (2,0.5) [anchor=north] {$0$};
\node at (3,0.5) [anchor=north] {$0$};
\node at (4,0.5) [anchor=north] {$0$};
\end{tikzpicture}
\end{equation}

\subsection{Computation of the partition function}\label{Sect.6.2}

The explicit form of the partition function of the lattice model reviewed in Section \ref{Sect.5.1} was conjectured in \cite{BBCG} and first proved in \cite{MSW}. In this subsection, we use the method outlined in Section \ref{Sect.1.2} and the concepts introduced in Section \ref{Sect.5} to give a new proof of this conjecture. The argument is considerably simpler than that of \cite{MSW}, and directly leads to the expression of the partition function (without the need to guess the formula as in \cite{MSW}).

First we set up some notations. Let $[\pm r]:=\{1,\overline{1},\cdots,r,\overline{r}\}$, and let $B_r$ be the hyperoctahedral group of degree $r$. We identify $\overline{\overline{i}}$ with $i$ for $1\leq i\leq r$. Any element $\sigma\in B_r$ can be viewed as a permutation of $[\pm r]$ such that $\sigma(\overline{i})=\overline{\sigma(i)}$ for any $1\leq i\leq r$. For any rational function $f(z)$ of $z=(z_1,\cdots,z_r)$, we define $\sigma f(z):=f(\sigma z)$, where $\sigma(z):=(z_{\sigma(1)},\cdots,z_{\sigma(r)})$ with $z_{\overline{i}}:=z_i^{-1}$ for every $1\leq i\leq r$.

The main result is the following theorem.

\begin{theorem}\label{Theorem2}
\begin{equation}
Z(\mathcal{T}_{\lambda,z})=z^{-\rho_B}\prod_{i=1}^r(1-\sqrt{v}z_i)\prod_{1\leq i<j\leq r}((1-vz_iz_j)(1-vz_jz_i^{-1}))\sum_{\sigma\in B_r}\sigma(z^{\lambda+\rho_C}\prod_{i=1}^r (1+\sqrt{v}z_i^{-1})\Delta_C(z)^{-1}),
\end{equation}
where $\rho_B=(r-\frac{1}{2},r-\frac{3}{2},\cdots,\frac{1}{2})$, $\rho_C=(r,r-1,\cdots,1)$, and 
\begin{equation*}
    \Delta_C(z)=\prod_{1\leq i<j\leq r}((z_i^{\frac{1}{2}}z_j^{-\frac{1}{2}}-z_i^{-\frac{1}{2}}z_j^{\frac{1}{2}})(z_i^{\frac{1}{2}}z_j^{\frac{1}{2}}-z_i^{-\frac{1}{2}}  z_j^{-\frac{1}{2}}))\prod_{i=1}^r (z_i-z_i^{-1}).
\end{equation*}
\end{theorem}

Throughout the rest of this section, we take $N=2r$, $\vec{x}=(z_1^{-1},z_1,\cdots,z_r^{-1},z_r)$ and $\vec{\epsilon}=(1,-1,\cdots,1,-1)$. For any $j\in \{1,2,\cdots,2r\}$, we denote by $x_j$ the $j$th entry of $\vec{x}$. Below we omit the arguments $\vec{x},\vec{\epsilon}$ from the notations $S^{[\alpha]}(\vec{x};\vec{\epsilon})$, $F(\vec{x};\vec{\epsilon})$, etc.

In order to treat the cap vertices, we define the ``cap vector'' $K\in  W_1\otimes \cdots \otimes W_N $ as follows. For every $\alpha,\beta\in\{0,1\}$, we denote by $C(\alpha,\beta;z)$ the Boltzmann weight of the cap of spectral parameter $z$ with $\alpha$ on the bottom edge and $\beta$ on the top edge. We define $K\in W_1\otimes \cdots \otimes W_N$ by specifying its component
\begin{equation}
    K_{i_1 \cdots i_N}:=\prod_{l=1}^r C(i_{2l-1},i_{2l};z_l)
\end{equation}
for every $(i_1,\cdots,i_N)\in \{0,1\}^N$.

Note that the partition function $Z(\mathcal{T}_{\lambda,z})$ can be written in terms of the column operators
\begin{equation}\label{E04}
  Z(\mathcal{T}_{\lambda,z})=\langle 0,\cdots,0|S^{[m_{\lambda_1+r-\frac{1}{2}}]}\cdots S^{[m_{\frac{1}{2}}]} K,
\end{equation}
where for every $j\in\{\frac{1}{2},\frac{3}{2},\cdots, \lambda_1+r-\frac{1}{2}\}$, $m_j=1$ if $j\in\{\lambda_i+r-i+\frac{1}{2}: i\in\{1,2,\cdots,r\}\}$ and $m_j=0$ otherwise.

Conjugating each column operator by the $F$-matrix, we write (\ref{E04}) in the following form
\begin{equation}\label{Par2}
    Z(\mathcal{T}_{\lambda,z})=(\langle 0,\cdots,0| F^{-1}) (F S^{[m_{\lambda_1+r-\frac{1}{2}}]}F^{-1}) \cdots (F S^{[m_{\frac{1}{2}}]}F^{-1}) (FK).
\end{equation}
Therefore, it suffices to compute $\langle 0,\cdots,0| F^{-1}$, the components of the conjugated column operator $F S^{[m_{j}]}F^{-1}$ for $j\in\{\frac{1}{2},\frac{3}{2},\cdots,\lambda_1+r-\frac{1}{2}\}$, and the components of $FK$.

By Proposition \ref{caduc}, we can deduce the following proposition.

\begin{proposition}\label{P4}
Assume that $t\in\{1,2,\cdots\}$ and $K_0,K_1,\cdots,K_t$ are cap vertices with spectral parameters $z_0,z_1,\cdots,z_t$. For any fixed $\alpha_0,\beta_0,\alpha_1,\beta_1,\cdots,\alpha_t,\beta_t\in \{0,1\}$ as indicated below, if $\alpha_0=\beta_0$, then the partition function of the following configuration is $0$. 
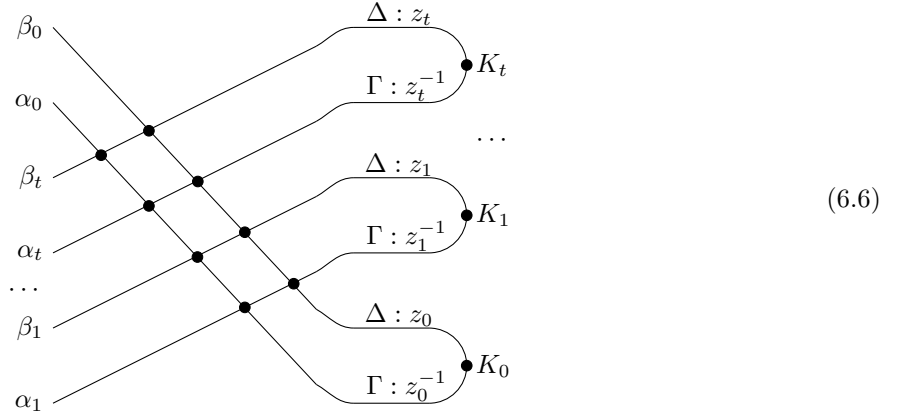
\begin{equation}
\hfill
\begin{tikzpicture}[baseline=(current bounding box.center)]
  \draw (0,0) to (3.5,1.75);
  \draw (3.5,1.75) to [out=30,in=180] (4,2);
  \draw (4,2) to (5,2);
  \draw (0,1)  to (3.5,2.75);
  \draw (3.5,2.75) to [out=30,in=180] (4,3);
  \draw (4,3) to (5,3);
  \draw (0,2) to (3.5,3.75);
  \draw (3.5,3.75) to [out=30,in=180] (4,4);
  \draw (4,4) to (5,4);
  \draw (0,3) to (3.5,4.75);
  \draw (3.5,4.75) to [out=30,in=180] (4,5);
  \draw (4,5) to (5,5);
  \draw (0,4) to (3.5,0.25);
   \draw (3.5,0.25) to [out=-30,in=180] (4,0);
  \draw (4,0) to (5,0);
  \draw (0,5) to (3.5,1.25);
  \draw (3.5,1.25) to [out=-30,in=180] (4,1);
  \draw (4,1) to (5,1);
  \draw (5,0) arc(-90:90:0.5);
  \draw (5,2) arc(-90:90:0.5);
  \draw (5,4) arc(-90:90:0.5);
  \filldraw[black] (5.5,0.5) circle (2pt);
  \filldraw[black] (5.5,2.5) circle (2pt);
  \filldraw[black] (5.5,4.5) circle (2pt);
  \node at (0,0) [anchor=east] {$\alpha_1$};
  \node at (0,1) [anchor=east] {$\beta_1$};
  \node at (0,2) [anchor=east] {$\alpha_t$};
  \node at (0,3) [anchor=east] {$\beta_t$};
  \node at (0,4) [anchor=east] {$\alpha_0$};
  \node at (0,5) [anchor=east] {$\beta_0$};
  \node at (0,1.5) [anchor=east] {$\cdots$};
  \node at (5.5,4.5) [anchor=west] {$K_t$};
  \node at (5.5,3.5) [anchor=west] {$\cdots$};
  \node at (5.5,0.5) [anchor=west] {$K_0$};
  \node at (5.5,2.5) [anchor=west] {$K_1$};
 \filldraw[black] (0.64,3.3) circle (2pt);
  \filldraw[black] (1.275,3.625) circle (2pt);
  \filldraw[black] (1.275,2.625) circle (2pt);
   \filldraw[black] (1.925,2.95) circle (2pt);
  \filldraw[black] (1.92,1.945) circle (2pt);
  \filldraw[black] (2.55,2.275) circle (2pt);
   \filldraw[black] (2.55,1.275) circle (2pt);
  \filldraw[black] (3.2,1.59) circle (2pt);
  \node at (4.7,-0.1) [anchor=south] {$\Gamma:z_0^{-1}$};
  \node at (4.6,0.9) [anchor=south] {$\Delta:z_0$};
  \node at (4.7,1.9) [anchor=south] {$\Gamma:z_1^{-1}$};
  \node at (4.6,2.9) [anchor=south] {$\Delta:z_1$};
  \node at (4.7,3.9) [anchor=south] {$\Gamma:z_t^{-1}$};
  \node at (4.6,4.9) [anchor=south] {$\Delta:z_t$};
\end{tikzpicture}
\end{equation}
\end{proposition}
\begin{proof}
Applying the caduceus relation (Proposition \ref{caduc}) to the caduceus braid at the bottom right corner, we deduce that the partition function is a constant times the partition function of the following
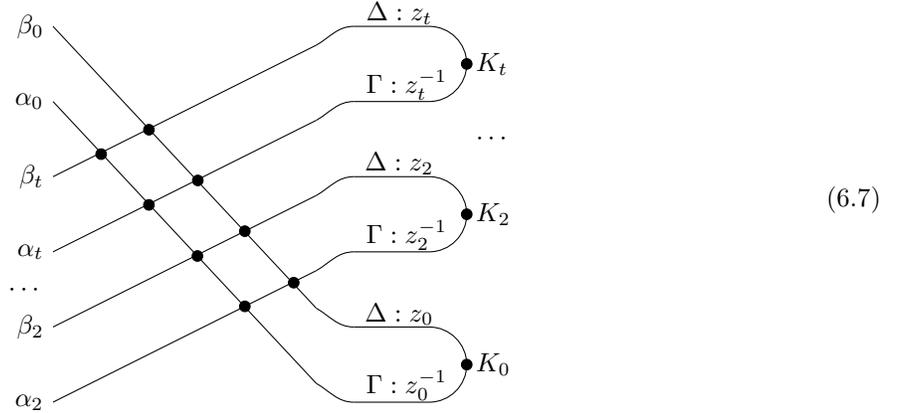
\begin{equation}
\hfill
\begin{tikzpicture}[baseline=(current bounding box.center)]
\draw (0,0) to (3.5,1.75);
  \draw (3.5,1.75) to [out=30,in=180] (4,2);
  \draw (4,2) to (5,2);
  \draw (0,1)  to (3.5,2.75);
  \draw (3.5,2.75) to [out=30,in=180] (4,3);
  \draw (4,3) to (5,3);
  \draw (0,2) to (3.5,3.75);
  \draw (3.5,3.75) to [out=30,in=180] (4,4);
  \draw (4,4) to (5,4);
  \draw (0,3) to (3.5,4.75);
  \draw (3.5,4.75) to [out=30,in=180] (4,5);
  \draw (4,5) to (5,5);
  \draw (0,4) to (3.5,0.25);
   \draw (3.5,0.25) to [out=-30,in=180] (4,0);
  \draw (4,0) to (5,0);
  \draw (0,5) to (3.5,1.25);
  \draw (3.5,1.25) to [out=-30,in=180] (4,1);
  \draw (4,1) to (5,1);
  \draw (5,0) arc(-90:90:0.5);
  \draw (5,2) arc(-90:90:0.5);
  \draw (5,4) arc(-90:90:0.5);
  \filldraw[black] (5.5,0.5) circle (2pt);
  \filldraw[black] (5.5,2.5) circle (2pt);
  \filldraw[black] (5.5,4.5) circle (2pt);
  \node at (0,0) [anchor=east] {$\alpha_2$};
  \node at (0,1) [anchor=east] {$\beta_2$};
  \node at (0,2) [anchor=east] {$\alpha_t$};
  \node at (0,3) [anchor=east] {$\beta_t$};
  \node at (0,4) [anchor=east] {$\alpha_0$};
  \node at (0,5) [anchor=east] {$\beta_0$};
  \node at (0,1.5) [anchor=east] {$\cdots$};
  \node at (5.5,4.5) [anchor=west] {$K_t$};
  \node at (5.5,3.5) [anchor=west] {$\cdots$};
  \node at (5.5,0.5) [anchor=west] {$K_0$};
  \node at (5.5,2.5) [anchor=west] {$K_2$};
 \filldraw[black] (0.64,3.3) circle (2pt);
  \filldraw[black] (1.275,3.625) circle (2pt);
  \filldraw[black] (1.275,2.625) circle (2pt);
   \filldraw[black] (1.925,2.95) circle (2pt);
  \filldraw[black] (1.92,1.945) circle (2pt);
  \filldraw[black] (2.55,2.275) circle (2pt);
   \filldraw[black] (2.55,1.275) circle (2pt);
  \filldraw[black] (3.2,1.59) circle (2pt);
  \node at (4.7,-0.1) [anchor=south] {$\Gamma:z_0^{-1}$};
  \node at (4.6,0.9) [anchor=south] {$\Delta:z_0$};
  \node at (4.7,1.9) [anchor=south] {$\Gamma:z_2^{-1}$};
  \node at (4.6,2.9) [anchor=south] {$\Delta:z_2$};
  \node at (4.7,3.9) [anchor=south] {$\Gamma:z_t^{-1}$};
  \node at (4.6,4.9) [anchor=south] {$\Delta:z_t$};
\end{tikzpicture}
\end{equation}

Repeating the procedure, we conclude that the original partition function is a constant times the Boltzmann weight of the following configuration, which is $0$ as $\alpha_0=\beta_0$. Hence the original partition function is $0$.
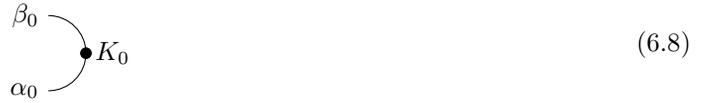
\begin{equation}
\hfill
\begin{tikzpicture}[baseline=(current bounding box.center)]
  \draw (0,0) arc(-90:90:0.5);
  \filldraw[black] (0.5,0.5) circle (2pt);
  \node at (0,0) [anchor=east] {$\alpha_0$};
  \node at (0,1) [anchor=east] {$\beta_0$};
  \node at (0.5,0.5) [anchor=west] {$K_0$};
\end{tikzpicture}
\end{equation}

\end{proof}

Based on Proposition \ref{P4}, we obtain the following proposition. It explicitly computes $ \langle   0,\cdots,0|F^{-1}$, the components of $F S^{[m]}F^{-1}$ for $m\in\{0,1\}$, and the components of $FK$.

\begin{proposition}\label{P5}
We have
\begin{equation}\label{EE1}
   \langle   0,\cdots,0|F^{-1}=\langle 0,\cdots,0|.
\end{equation}
For any $(i_1,\cdots,i_N),(j_1,\cdots,j_N)\in\{0,1\}^N$, we have
\begin{equation}\label{EE2}
    (F S^{[0]}  F^{-1})_{i_1\cdots i_N}^{j_1\cdots j_N}=\prod_{t=1}^N \mathbbm{1}_{i_t=j_t}\prod_{t: i_t=1}b_2(x_t;\epsilon_t),
\end{equation}
\begin{eqnarray}\label{EE3}
       (F S^{[1]} F^{-1})_{i_1\cdots i_N}^{j_1\cdots j_N}&=&\sum_{m=1}^N \mathbbm{1}_{i_m=0,j_m=1}\prod_{t:1\leq t\leq N, t\neq m}\mathbbm{1}_{i_t=j_t}\prod_{t: i_t=1,j_t=1}b_2(x_t;\epsilon_t)\nonumber\\
    &&\times  \prod_{t:t>m,i_t=1,j_t=1}a_2^{-1}(x_m,x_t;\epsilon_m,\epsilon_t)\prod_{t:i_t=0,j_t=0}b_2^{-1}(x_m,x_t;\epsilon_m,\epsilon_t),
\end{eqnarray}
\begin{eqnarray}\label{EE4}
(FK)_{i_1,\cdots,i_{N}} &=& \prod_{a=1}^r \mathbbm{1}_{i_{2a-1}\neq i_{2a}} \prod_{a=1}^r z_a^{-\frac{1}{2}}\prod_{a:1\leq a\leq r, i_{2a-1}=0,i_{2a}=1}\frac{z_a+\sqrt{v}}{z_a^{-1}+\sqrt{v}}\nonumber\\
&\times& \prod_{(a,b):1\leq a<b\leq r, i_{2a-1}=1,i_{2b-1}=1}b_2(z_b^{-1},z_a;1,-1)\prod_{(a,b):1\leq a<b\leq r,i_{2a-1}=1,i_{2b}=1}b_2(z_b,z_a;-1,-1)\nonumber\\
&\times& \prod_{(a,b):1\leq a<b\leq r, i_{2a}=1,i_{2b-1}=1}b_2(z_b^{-1},z_a^{-1};1,1)\prod_{(a,b):1\leq a<b\leq r,i_{2a}=1,i_{2b}=1}b_2(z_b,z_a^{-1};-1,1).
\end{eqnarray}
Equivalently, we have
\begin{equation}
    F S^{[0]} F^{-1}=\bigotimes_{t\in\{1,\cdots,N\}}\begin{pmatrix}
    1 & 0\\
    0 & b_2(x_t;\epsilon_t)
    \end{pmatrix}_t,
\end{equation}
\begin{eqnarray}
     F S^{[1]} F^{-1}=\sum_{m=1}^N &&\bigotimes_{t\in \{1,\cdots,m-1\}}
    \begin{pmatrix}
    b_2^{-1}(x_m,x_t;\epsilon_m,\epsilon_t) & 0\\
    0 & b_2(x_t;\epsilon_t)
    \end{pmatrix}_t
    \bigotimes \begin{pmatrix}
    0 & 1\\
    0 & 0
    \end{pmatrix}_m\nonumber\\
    &&    \bigotimes_{t\in\{m+1,\cdots,N\}}
    \begin{pmatrix}
    b_2^{-1}(x_m,x_t;\epsilon_m,\epsilon_t) & 0\\
    0 & \frac{b_2(x_t;\epsilon_t)}{a_2(x_m,x_t;\epsilon_m,\epsilon_t)}
    \end{pmatrix}_t,
\end{eqnarray}
\begin{eqnarray}
    FK=(\prod_{a=1}^r z_a^{-\frac{1}{2}})\sum_{(e_1,\cdots,e_r)\in \{\pm 1\}^r}&&\prod_{a=1}^r\frac{z_a^{-e_a}+\sqrt{v}}{z_a^{-1}+\sqrt{v}}  \prod_{1\leq a<b\leq r} b_2(z_b^{-e_b},z_a^{e_a};e_b,-e_a)\nonumber\\ &&\bigotimes_{a\in\{1,\cdots,r\}}|e_a\rangle_{2a-1}\otimes |-e_a\rangle_{2a},
\end{eqnarray}
where the basis vectors of each $W_t$ for $t\in\{1,\cdots,N\}$ are ordered as $|0\rangle$, $|1\rangle$, and for every $1\leq t\leq N$,
\begin{equation*}
    |+1\rangle_t=\begin{pmatrix}
    0\\
    1
    \end{pmatrix}_t ,  \quad |-1\rangle_t=\begin{pmatrix}
    1\\
    0
    \end{pmatrix}_t. 
\end{equation*}
\end{proposition}
\begin{proof}

The results (\ref{EE1}), (\ref{EE2}), and (\ref{EE3}) can be proved in a similar manner as Proposition \ref{P3}. Below we present the proof of (\ref{EE4}).

We note that by an analog of Proposition \ref{P2}, 
\begin{equation}
    (FK)_{i_1\cdots i_N}=(R_{id}^{\rho}K)_{i_1\cdots i_N},
\end{equation}
where $\rho\in S_N$ is the unique permutation that satisfies the condition 
\begin{eqnarray}
  && 0\leq i_{\rho(N)}\leq \cdots \leq i_{\rho(1)} \leq 1,\nonumber\\
  && i_{\rho(t)}=i_{\rho(t+1)}\text{ implies }\rho(t)<\rho(t+1), \text{ for every }  1\leq t\leq N-1.
\end{eqnarray}

By spin conservation and the Boltzmann weights of the caps, in order for $(R_{id}^{\rho} K)_{i_1\cdots i_N}$ to be non-vanishing, we must have
\begin{eqnarray*}
    &&i_{\rho(1)}=\cdots=i_{\rho(r)}=1,\quad
     i_{\rho(r+1)}=\cdots=i_{\rho(2r)}=0,\\
    && \rho(1)<\cdots<\rho(r),\quad \rho(r+1)<\cdots<\rho(2r).
\end{eqnarray*}

Note that $(R_{id}^{\rho} K)_{i_1\cdots i_N}$ is the partition function of a lattice model. An illustration is given below.
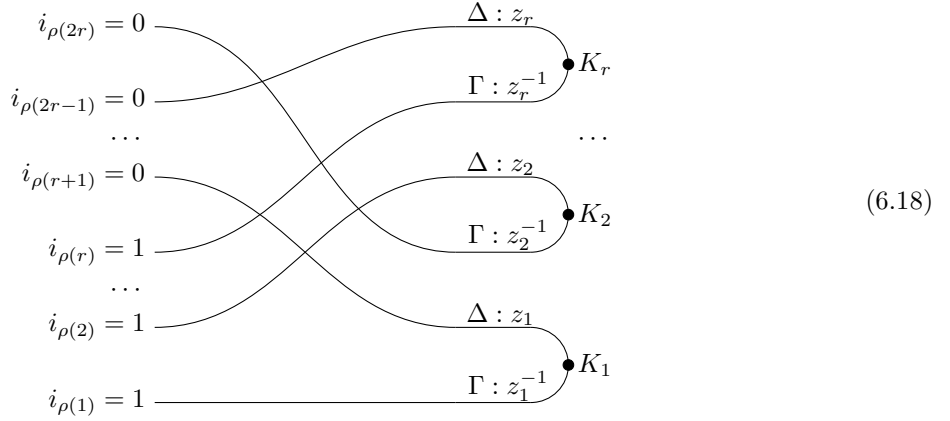
\begin{equation}
\hfill
\begin{tikzpicture}[baseline=(current bounding box.center)]
  \draw (4,2) to (5,2);
  \draw (4,3) to (5,3);
  \draw (4,4) to (5,4);
  \draw (4,5) to (5,5);
  \draw (4,0) to (5,0);
  \draw (4,1) to (5,1);
  \draw (0,0) to (4,0);
  \draw (0,3) to [out=0,in=180] (4,1);
  \draw (0,1) to [out=0, in=180] (4,3);
  \draw (0,5)  to [out=0,in=180] (4,2);
  \draw (0,2) to [out=0,in=180] (4,4);
  \draw (0,4) to [out=0,in=180] (4,5);
  \draw (5,0) arc(-90:90:0.5);
  \draw (5,2) arc(-90:90:0.5);
  \draw (5,4) arc(-90:90:0.5);
  \filldraw[black] (5.5,0.5) circle (2pt);
  \filldraw[black] (5.5,2.5) circle (2pt);
  \filldraw[black] (5.5,4.5) circle (2pt);
  \node at (0,0) [anchor=east] {$i_{\rho(1)}=1$};
  \node at (0,1) [anchor=east] {$i_{\rho(2)}=1$};
  \node at (0,3) [anchor=east] {$i_{\rho(r+1)}=0$};
  \node at (0,2) [anchor=east] {$i_{\rho(r)}=1$};
  \node at (0,4) [anchor=east] {$i_{\rho(2r-1)}=0$};
  \node at (0,5) [anchor=east] {$i_{\rho(2r)}=0$};
  \node at (0,3.5) [anchor=east] {$\cdots$};
  \node at (0,1.5) [anchor=east] {$\cdots$};
  \node at (5.5,4.5) [anchor=west] {$K_r$};
  \node at (5.5,3.5) [anchor=west] {$\cdots$};
  \node at (5.5,0.5) [anchor=west] {$K_1$};
  \node at (5.5,2.5) [anchor=west] {$K_2$};
  \node at (4.7,-0.1) [anchor=south] {$\Gamma:z_1^{-1}$};
  \node at (4.6,0.9) [anchor=south] {$\Delta:z_1$};
  \node at (4.7,1.9) [anchor=south] {$\Gamma:z_2^{-1}$};
  \node at (4.6,2.9) [anchor=south] {$\Delta:z_2$};
   \node at (4.7,3.9) [anchor=south] {$\Gamma:z_{r}^{-1}$};
  \node at (4.6,4.9) [anchor=south] {$\Delta:z_{r}$};
\end{tikzpicture}
\end{equation}
By spin conservation, the Boltzmann weights of the cap vertices, and Proposition \ref{P4}, we can deduce that for any admissible state of the above lattice model, the situation that $i_{1}=i_{2}=1$ or $i_{1}=i_{2}=0$ cannot happen. Hence for any admissible state, either $i_{1}=1,i_{2}=0$, or $i_{1}=0,i_{2}=1$. For either case, by a similar argument as that in the proof of Proposition \ref{P3}, we can remove the cap vertex $K_1$ together with the two lines associated with it up to a constant factor. Repeating the above argument, we can deduce that for any admissible state and any $1\leq a\leq r$, either $i_{2a-1}=1,i_{2a}=0$, or $i_{2a-1}=0,i_{2a}=1$.



Note that using a similar argument as that in the proof of Proposition \ref{P3}, if $i_{1}=1,i_2=0$, we can remove the cap vertex $K_1$ together with the two lines associated with it up to a factor of
\begin{equation}
    z_1^{-\frac{1}{2}}\prod_{a:1<a\leq r, i_{2a-1}=1,i_{2a}=0}b_2(z_a^{-1},z_1;1,-1)\prod_{a:  1<a\leq r,i_{2a-1}=0,i_{2a}=1}b_2(z_a,z_1;-1,-1).
\end{equation}

Now note that the partition function of the following configuration can be computed as
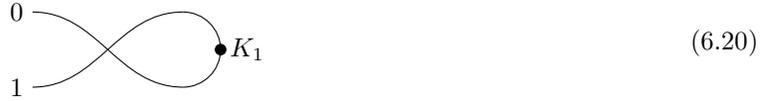
\begin{equation}
\hfill
\begin{tikzpicture}[baseline=(current bounding box.center)]
  \draw (0,-0.5) to  [out=0,in=180] (2,0.5);
  \draw (0,0.5) to  [out=0,in=180] (2,-0.5);
  \draw (2,-0.5) arc(-90:90:0.5);
  \filldraw[black] (2.5,0) circle (2pt);
  \node at (0,-0.5) [anchor=east] {$1$};
  \node at (0,0.5) [anchor=east] {$0$};
  \node at (2.5,0) [anchor=west] {$K_1$};
\end{tikzpicture}
\end{equation}
\begin{eqnarray}
  &&  c_1(z_1,z_1^{-1};-1,1)C(1,0;z_1)+b_2(z_1,z_1^{-1};-1,1)C(0,1;z_1)\nonumber \\
  &=& \frac{(1-v)z_1}{z_1^{-1}-vz_1} z_1^{-\frac{1}{2}}-\frac{z_1-z_1^{-1}}{z_1^{-1}-vz_1}\sqrt{v} z_1^{\frac{1}{2}}=\frac{z_1^{-\frac{1}{2}}(z_1+\sqrt{v})}{z_1^{-1}+\sqrt{v}}.
\end{eqnarray}
Therefore, if $i_1=0,i_2=1$, we can remove the cap $K_1$ together with the two lines associated with it up to a factor of
\begin{equation}
    \frac{z_1^{-\frac{1}{2}}(z_1+\sqrt{v})}{z_1^{-1}+\sqrt{v}}\prod_{a:1<a\leq r, i_{2a-1}=1,i_{2a}=0}b_2(z_a^{-1},z_1^{-1};1,1)\prod_{a:  1<a\leq r,i_{2a-1}=0,i_{2a}=1}b_2(z_a,z_1^{-1};-1,1).
\end{equation}

By repeatedly removing each cap and its two associated lines (from bottom to top), we conclude that
\begin{eqnarray}
(FK)_{i_1,\cdots,i_{N}} &=& \prod_{a=1}^r \mathbbm{1}_{i_{2a-1}\neq i_{2a}} \prod_{i=1}^r z_i^{-\frac{1}{2}}\prod_{a:1\leq a\leq r, i_{2a-1}=0,i_{2a}=1}\frac{z_a+\sqrt{v}}{z_a^{-1}+\sqrt{v}}\nonumber\\
&\times& \prod_{(a,b):1\leq a<b\leq r, i_{2a-1}=1,i_{2b-1}=1}b_2(z_b^{-1},z_a;1,-1)\prod_{(a,b):1\leq a<b\leq r,i_{2a-1}=1,i_{2b}=1}b_2(z_b,z_a;-1,-1)\nonumber\\
&\times& \prod_{(a,b):1\leq a<b\leq r, i_{2a}=1,i_{2b-1}=1}b_2(z_b^{-1},z_a^{-1};1,1)\prod_{(a,b):1\leq a<b\leq r,i_{2a}=1,i_{2b}=1}b_2(z_b,z_a^{-1};-1,1).
\end{eqnarray}

\end{proof}

Now we finish the proof of Theorem \ref{Theorem2} using Proposition \ref{P5}.

\begin{proof}[Proof of Theorem \ref{Theorem2}]
Note that by Proposition \ref{P5} and (\ref{Par2}), $Z(\mathcal{T}_{\lambda,z})$ can be written as the sum of the following terms
\begin{equation}
    (F S^{[m_{\lambda_1+r-\frac{1}{2}}]} F^{-1})_{0\cdots 0}^{i_1^{(\lambda_1+r-\frac{1}{2})}\cdots i_N^{(\lambda_1+r-\frac{1}{2})}}\cdots (F S^{[m_{\frac{3}{2}}]} F^{-1})_{i_1^{(\frac{5}{2})}\cdots i_N^{(\frac{5}{2})}}^{i_1^{(\frac{3}{2})}\cdots i_N^{(\frac{3}{2})}}(F S^{[m_{\frac{1}{2}}]} F^{-1})_{i_1^{(\frac{3}{2})}\cdots i_N^{(\frac{3}{2})}}^{i_1^{(\frac{1}{2})}  \cdots i_N^{(\frac{1}{2})}}(FK)_{i_1^{(\frac{1}{2})}\cdots i_N^{(\frac{1}{2})}},
\end{equation}
where $(i_1^{(l)},\cdots,i_N^{(l)})\in\{0,1\}^N$ for $l=\frac{1}{2},\frac{3}{2},\cdots, \lambda_1+r-\frac{1}{2}$ satisfy the following condition: if $m_l=0$, then $i_k^{(l)}=i_k^{(l+1)}$ for every $1\leq k\leq N$; if $m_l=1$, then there is a unique index $\alpha_l\in\{1,2,\cdots,r\}$ and a unique integer $e_l\in\{0,1\}$, such that $i_{2\alpha_l-e_l}^{(l)}=1,i_{2\alpha_l-e_l}^{(l+1)}=0,i_{2\alpha_l-1+e_l}^{(l)}=0,i_{2\alpha_l-1+e_l}^{(l+1)}=0$, and $i_k^{(l)}=i_k^{(l+1)}$ for every $k\neq 2\alpha_l-1,2\alpha_l$. Here, we have assumed that $(i_1^{(\lambda_1+r+\frac{1}{2})},\cdots,i_N^{(\lambda_1+r+\frac{1}{2})})=(0,\cdots,0)$.

Let $\beta_t:=\alpha_{\lambda_t+r-t+\frac{1}{2}}$, $f_t:=2e_t-1$ for every $1\leq t\leq r$. Let $f:=(f_1,\cdots,f_r)\in\{\pm 1\}^r $. Note that $(\beta_1,\cdots,\beta_r)$ corresponds to a permutation $\sigma\in S_r$ such that $\sigma(t)=\beta_t$ for every $1\leq t\leq r$. Based on this observation and Proposition \ref{P5}, we can deduce that
\begin{eqnarray*}
Z(\mathcal{T}_{\lambda,z}) &=& \sum_{\sigma\in S_r}\sum_{f\in \{\pm 1\}^r}\prod_{i=1}^r z_{\sigma(i)}^{-f_i (\lambda_i+r-i)}\prod_{i=1}^r z_i^{-\frac{1}{2}}\prod_{i=1}^r\frac{z_{\sigma(i)}^{-f_i}+\sqrt{v}}{z_i^{-1}+\sqrt{v}}\prod_{i=1}^r\frac{z_i^{-1}-vz_i}{z_{\sigma(i)}^{-f_i}-z_{\sigma(i)}^{f_i}}\prod_{1\leq i<j\leq r}B(i,j,\sigma,f),
\end{eqnarray*}
where $B(i,j,\sigma,f)$ is given as follows. If $\sigma(i)<\sigma(j)$, then
\begin{equation}
    B(i,j,\sigma,f):=b_2^{-1}(z_{\sigma(i)}^{-f_i},z_{\sigma(j)}^{-f_j};f_i,f_j) b_2^{-1}(z_{\sigma(i)}^{-f_i},z_{\sigma(j)}^{f_j}; f_i,-f_j);
\end{equation}
if $\sigma(i)>\sigma(j)$, then
\begin{equation}
    B(i,j,\sigma,f):=a_2^{-1}(z_{\sigma(j)}^{-f_j},z_{\sigma(i)}^{-f_i};f_j,f_i)b_2^{-1}(z_{\sigma(i)}^{-f_i},z_{\sigma(j)}^{-f_j};f_i,f_j) b_2^{-1}(z_{\sigma(j)}^{-f_j},z_{\sigma(i)}^{f_i}; f_j,-f_i).
\end{equation}

By computation, we obtain that if $\sigma(i)<\sigma(j)$, 
\begin{eqnarray*}
 B(i,j,\sigma,f)&=& (z_{\sigma(i)}^{\frac{1}{2}}z_{\sigma(j)}^{-\frac{1}{2}}-v z_{\sigma(i)}^{-\frac{1}{2}} z_{\sigma(j)}^{\frac{1}{2}})
 (z_{\sigma(i)}^{-\frac{1}{2}}z_{\sigma(j)}^{-\frac{1}{2}}-v z_{\sigma(i)}^{\frac{1}{2}} z_{\sigma(j)}^{\frac{1}{2}})\\
 &&\times (z_{\sigma(i)}^{-\frac{1}{2}f_i}z_{\sigma(j)}^{\frac{1}{2}f_j}-z_{\sigma(i)}^{\frac{1}{2}f_i}z_{\sigma(j)}^{-\frac{1}{2}f_j})^{-1}(z_{\sigma(i)}^{-\frac{1}{2}f_i}z_{\sigma(j)}^{-\frac{1}{2}f_j}-z_{\sigma(i)}^{\frac{1}{2}f_i}z_{\sigma(j)}^{\frac{1}{2}f_j})^{-1};
\end{eqnarray*}
if $\sigma(i)>\sigma(j)$,
\begin{eqnarray*}
 B(i,j,\sigma,f)&=& (z_{\sigma(i)}^{-\frac{1}{2}}z_{\sigma(j)}^{\frac{1}{2}}-v z_{\sigma(i)}^{\frac{1}{2}} z_{\sigma(j)}^{-\frac{1}{2}})
 (z_{\sigma(i)}^{-\frac{1}{2}}z_{\sigma(j)}^{-\frac{1}{2}}-v z_{\sigma(i)}^{\frac{1}{2}} z_{\sigma(j)}^{\frac{1}{2}})\\
 &&\times (z_{\sigma(i)}^{-\frac{1}{2}f_i}z_{\sigma(j)}^{\frac{1}{2}f_j}-z_{\sigma(i)}^{\frac{1}{2}f_i}z_{\sigma(j)}^{-\frac{1}{2}f_j})^{-1}(z_{\sigma(i)}^{-\frac{1}{2}f_i}z_{\sigma(j)}^{-\frac{1}{2}f_j}-z_{\sigma(i)}^{\frac{1}{2}f_i}z_{\sigma(j)}^{\frac{1}{2}f_j})^{-1}.
\end{eqnarray*}
Thus we have
\begin{eqnarray*}
\prod_{1\leq i<j\leq r}B(i,j,\sigma,f) &=& \prod_{1\leq i<j\leq r}\frac{(z_i^{\frac{1}{2}}z_j^{-\frac{1}{2}}-vz_i^{-\frac{1}{2}}z_j^{\frac{1}{2}})(z_i^{-\frac{1}{2}}z_j^{-\frac{1}{2}}-vz_i^{\frac{1}{2}}z_j^{\frac{1}{2}})}{(z_{\sigma(i)}^{-\frac{1}{2}f_i}z_{\sigma(j)}^{\frac{1}{2}f_j}-z_{\sigma(i)}^{\frac{1}{2}f_i}z_{\sigma(j)}^{-\frac{1}{2}f_j})(z_{\sigma(i)}^{-\frac{1}{2}f_i}z_{\sigma(j)}^{-\frac{1}{2}f_j}-z_{\sigma(i)}^{\frac{1}{2}f_i}z_{\sigma(j)}^{\frac{1}{2}f_j})}.
\end{eqnarray*}

Therefore, we conclude that
\begin{equation}
Z(\mathcal{T}_{\lambda,z})=z^{-\rho_B}\prod_{i=1}^r(1-\sqrt{v}z_i)\prod_{1\leq i<j\leq r}((1-vz_iz_j)(1-vz_jz_i^{-1}))\sum_{\sigma\in B_r}\sigma(z^{\lambda+\rho_C}\prod_{i=1}^r (1+\sqrt{v}z_i^{-1})\Delta_C(z)^{-1}).
\end{equation}

\end{proof}

\newpage
\bibliographystyle{acm}
\bibliography{Partition.bib}

\end{document}